\documentclass[10pt,twocolumn]{IEEEtran}%

\usepackage{amsmath}
\usepackage{amsfonts}
\usepackage{amssymb}
\usepackage{graphicx}
\usepackage{psfrag}
\usepackage{cite}
\usepackage{algorithm}
\usepackage{color}
\usepackage{float}
\usepackage{epsfig,array,multicol,verbatim,algorithmic}%
\IEEEoverridecommandlockouts
\newtheorem{theorem}{\bf Theorem}
\newtheorem{lemma}{\bf Lemma}
\newtheorem{property}{\bf Property}

\newtheorem{definition}{\bf Definition}
\newtheorem{convention}{\bf Convention}
\newtheorem{remark}{\bf Remark}
\newlength{\aligntop}
\newlength{\alignbot}
\makeatletter
\renewenvironment{align}{%
  \vspace{\aligntop}
  \start@align\@ne\st@rredfalse\m@ne
}{%
  \math@cr \black@\totwidth@
  \egroup
  \ifingather@
    \restorealignstate@
    \egroup
    \nonumber
    \ifnum0=`{\fi\iffalse}\fi
  \else
    $$%
  \fi
  \ignorespacesafterend%
  \vspace{\alignbot}\par\noindent
}
\newcommand{\mysection}[1]{\vspace*{-0.47em}\section{#1}\vspace*{-0.49em}}
\newcommand{\mysubsection}[1]{\vspace*{-1.05em}\subsection{#1}\vspace*{-0.2em}}

\begin{document}

\title{Coalition Formation Games for Collaborative Spectrum Sensing}
\author{Walid Saad, Zhu Han, Tamer Ba\c{s}ar, M\'{e}rouane Debbah, and Are Hj{\o}rungnes \vspace{-1.1cm} \thanks{Copyright (c) 2010 IEEE. Personal use of this material is permitted. However, permission to use this material for any other purposes must be obtained from the IEEE by sending a request to pubs-permissions@ieee.org.

W.~Saad and A. Hj{\o}rungnes are
with the UNIK Graduate University Center, University of Oslo, Oslo, Norway, e-mails: \texttt{\{saad,arehj\}@unik.no}.
Z.~Han is with Electrical and Computer Engineering Department, University of Houston, Houston, Tx, USA, email: \texttt{zhan2@mail.uh.edu}. T. Ba\c{s}ar is with the Coordinated Science Laboratory, University of Illinois at Urbana Champaign, IL, USA, email: \texttt{basar1@illinois.edu}. M.~Debbah is the Alcatel-Lucent chair, SUPELEC, Paris, France e-mail:
\texttt{merouane.debbah@supelec.fr}. This research is supported by the Research Council of Norway
 through the projects 183311/S10, 176773/S10 and 18778/V11, and through a MURI Grant at UIUC. A preliminary version of this work [50] appeared in the Proceedings of the IEEE Conference on Computer Communications (INFOCOM), April, 2009. }}
\maketitle

\begin{abstract}
Collaborative Spectrum Sensing (CSS) between secondary users (SUs) in cognitive networks exhibits an inherent tradeoff between minimizing the probability of missing the detection of the primary user (PU) and maintaining a reasonable false alarm probability (e.g., for maintaining a good spectrum utilization). In this paper, we study the impact of this tradeoff on the network structure and the cooperative incentives of the SUs that seek to cooperate for improving their detection performance. We model the CSS problem as a non-transferable coalitional game, and we propose distributed algorithms for coalition formation. First, we construct a  distributed coalition formation~(CF) algorithm that allows the SUs to self-organize into disjoint coalitions while accounting for the CSS tradeoff. Then, the CF algorithm is complemented with a coalitional voting game for enabling distributed coalition formation with detection probability guarantees~(CF-PD) when required by the PU. The CF-PD algorithm allows the SUs to form minimal winning coalitions (MWCs), i.e., coalitions that achieve the target detection probability with minimal costs. For both algorithms, we study and prove various properties pertaining to network structure, adaptation to mobility and stability. Simulation results show that CF reduces the average probability of miss per SU up to $88.45\%$ relative to the non-cooperative case, while maintaining a desired false alarm. For CF-PD, the results show that up to $87.25\%$ of the SUs achieve the required detection probability through MWCs.

%For both DCF and DCF-PD, we study and prove various properties and characteristics pertaining to coalition size, topology and stability.
%Collaborative spectrum sensing in cognitive networks, allows secondary users~(SUs) to improve their performance, .  However, there exists an inherent trade off between

%by using simple merge and split rules from coalitional games,

%Moreover, while existing literature mainly focused on centralized solutions for collaborative sensing, we propose distributed collaboration strategies through coalitional game theory.

%In addition, through simulations, we compare the performance of the proposed distributed solution with respect to an optimal centralized solution that minimizes the average missing probability per SU.  Finally, the results also show how the proposed algorithm autonomously adapts the network topology to environmental changes such as mobility.

% We study the stability of the resulting network structure, and show that a maximum number of SUs per formed coalition exists for the proposed utility model. F % or the introduction of new SUs.
 %WO ALGORITHMS FOR COALITION FORMATION, MAXIMIZE + CONSTRAINT, BOTH CONSTRAINTS DETECTOPM MISSING, MAXIMUM MISSING INSTEAD OF MINIMUM DETECTIon. ADD FEW WORDS ON COMPLEXITY
 %seeking to reduce the interference on the PU through collaborative sensing. T
\end{abstract}
%{\bf Keywords:} cognitive radio, coalitional games, collaborative sensing, coalition formation, game theory.
\vspace{-0.1cm}
\mysection{Introduction}
 Recently, there has been a surge in wireless services, yielding a huge demand on the radio spectrum. However, the spectrum resources are scarce and most of them have been already licensed to existing operators. Various studies have shown that the actual licensed spectrum remains unoccupied for large periods \cite{FCC}. For efficiently exploiting these spectrum holes, \emph{cognitive radio} has been proposed \cite{CR01}. By monitoring and adapting to the environment, cognitive radios (secondary users~(SUs)) can share the spectrum with the licensed user (primary user~(PU)), operating whenever the PU is not using the spectrum. However, deploying and implementing such flexible SUs faces several challenges \cite{CR02} at different levels such as spectrum access \cite{CR02,NEW02,NEW08} and spectrum sensing \cite{DT00,CS00,NEW04,CS01,CS02, SA00,CS04,CS03,CS05,CS06}. Spectrum sensing mainly deals with the stage during which the cognitive users attempt to learn their environment prior to accessing and sharing the spectrum. For instance, the SUs must constantly sense the spectrum in order to detect the presence of the PU and use the spectrum holes without causing harmful interference to the PU. Hence, efficient and distributed spectrum sensing is of utter importance and constitutes a major challenge in cognitive radio networks.

For sensing the PU presence, the SUs must be able to detect its signal using various kinds of detectors such as energy detectors or others \cite{DT00}. However, the performance of spectrum sensing is significantly affected by the degradation of the PU signal due to path loss or shadowing (hidden terminal). It has been shown that, through collaborative spectrum sensing~(CSS) among SUs, the effects of this hidden terminal problem can be reduced and the probability of detecting the PU can be improved \cite{CS00,NEW04,CS01,CS04,CS03,CS02,CS05,CS06}. In \cite{CS00}, the SUs collaborate by sharing their sensing decisions through a centralized fusion center which combines the SUs' sensing bits using the OR-rule for data fusion. A similar approach is used in \cite{CS01} using different decision-combining methods. In \cite{CS02}, it is shown that, in CSS, soft decisions can have an almost comparable performance with hard decisions while reducing complexity. The authors in \cite{CS04} propose an evolutionary game model for CSS in order to inspect the strategies of the SUs and their contribution to the sensing process. The effect of the sensing time on the access performance of the SUs in a cognitive network is analyzed in \cite{SA00}. For improving the performance of CSS, spatial diversity techniques are presented in \cite{CS03} as a means for combatting the error probability due to fading on the reporting channel between the SUs and the central fusion center. Other interesting performance aspects of spectrum sensing and CSS are studied in \cite{NEW04,CS07,YY01,YY02,CS05,CS06}. Existing literature mainly focused on the performance assessment of CSS in the presence of a centralized fusion center that combines \emph{all} the SUs bits. In practice, the SUs can be at different locations in the network, and, thus, prefer to form nearby groups for CSS without relying on a centralized entity. Moreover, the SUs can belong to different service providers and need to interact with each other for CSS, instead of relaying their bits to a centralized fusion center which may not even exist in an \emph{ad hoc} network of SUs. In addition, a centralized approach requires the deployment and/or re-use of an infrastructure for the purpose of spectrum sensing, which is quite impractical as many secondary networks are inherently ad hoc. Further, the use of a centralized spectrum sensing approach can yield a significant overhead and complexity, notably in networks with a large number of SUs having varying locations.

The main contribution of this paper is to devise distributed strategies for CSS between the SUs in a cognitive network which capture the spectrum sensing objectives of these SUs. A further contribution of this paper is to inspect how the network structure can evolve and change as a result of the cooperative incentives of the SUs seeking to improve their spectrum sensing performance, given the inherent tradeoff that exists between the CSS gains in terms of reducing the probability of missing the detection of the PU and the cooperation costs in terms of maintaining a reasonable false alarm probability (e.g., to maintain an efficient spectrum utilization). We model the CSS problem as a non-transferable coalitional game and we propose a distributed coalition formation~(CF) algorithm. Using CF, each SU autonomously decides to form or leave a coalition while maximizing its utility in terms of detection probability and accounting for the false alarm cost. We show that, due to the cost for cooperation, independent disjoint coalitions will form in the network and a maximum coalition size exists for the proposed utility model. For cognitive networks where the PU imposes a detection probability level on the SUs, we propose a distributed coalition formation algorithm with detection probability guarantees (CF-PD) which is built by overlaying a coalitional voting game on the CF algorithm. The CF-PD algorithm enables the SUs to autonomously take the decision to form minimal winning coalitions (MWCs), i.e., coalitions that achieve the target detection probability with minimal costs. Through simulations, we assess the performance of both algorithms relative to the non-cooperative case, we compare the results with a centralized solution, and we show how the SUs can self-organize and adapt the network structure to environmental changes such as mobility.% This trade off can be pictured as a trade off between reducing the interference on the PU (increasing the detection probability) while maintaining a good spectrum utilization (reducing the false alarm probability). F

Note that, although some existing literature such as \cite{CZ00} and \cite{HK00} have dealt with the idea of forming clusters for CSS, the objectives of \cite{CZ00} and \cite{HK00} are different from the contributions of our paper. For instance, \cite{CZ00} studies, in the context of the centralized fusion center based CSS, the possibility of using clusters for combatting the \emph{errors on the reporting channel}. However, the authors do not deal with the idea of distributed CSS and, as clearly stated in \cite[Section~III]{CZ00}, the paper does \emph{not propose any algorithm for forming the clusters}, instead it assumes that the clusters are pre-formed. Moreover, \cite{CZ00} does not account for the tradeoff between the probabilities of miss and false alarm, which impacts the size and identities of the SUs that will form coalitions or clusters. In this context, our paper complements \cite{CZ00} by having a number of different and novel contributions as discussed in the previous paragraph.

Further, in \cite{HK00}, the authors focus on studying the performance of CSS in a geographical area, using different detection techniques such as energy detection or feature detection. The idea of clustering in \cite{HK00} is restricted to dividing, based on the received signal strength of the PU, a large geographical area into a number of defined clusters which are geographical areas with a large radius, i.e., analogous to the cells of a wireless cellular system. Within every cluster or cell, the work in \cite{HK00} considers that a large number of SUs exist and these SUs utilize the traditional network-wide fusion center solution at cell-level.  Moreover, in \cite{HK00}, no target probability of miss or false alarm constraints are taken into account. In short, \cite{HK00} does not look at distributed strategies enabling the SUs to, autonomously, form coalitions while optimizing their probability of miss given a false alarm constraint. In essence, our approach complements the work done in \cite{HK00}. For instance, given a cognitive network deployed in a large geographical area, one can first cluster it geographically into cells using the approach \cite{HK00} and, then, our proposed coalition formation approach can be applied by the SUs inside each cluster or cell.

The rest of this paper is organized as follows: Section~\ref{sec:prob} presents
the system model. In Section \ref{sec:coal}, we present the proposed coalitional game and prove different properties while in Section~\ref{sec:coalform} we devise a distributed algorithm for autonomous coalition formation. Section~\ref{sec:dcfpd} extends the coalition formation algorithm to accommodate probability of detection guarantees. In Section~\ref{sec:compl}, we discuss implementation aspects in practical cognitive networks. Simulation results are presented and analyzed in Section \ref{sec:sim}. Finally, conclusions are drawn in
Section \ref{sec:conc}.

\mysection{System Model}\label{sec:prob}
Consider a network consisting of $N$ transmit-receive pairs of SUs and a single PU with $\mathcal{N} = \{1,\ldots,N\}$ denoting the set of all SUs. The SUs and the PU can be either stationary or mobile. Throughout this paper, we assume that all SUs are honest nodes with no malicious or cheating behavior. Since the focus is on spectrum sensing, we are only interested in the transmitter part of each of the $N$ SUs. In a non-cooperative approach, each of the $N$ SUs continuously senses the spectrum in order to detect the presence of the PU. For detecting the PU, we use energy detectors which are one of the main practical signal detectors in cognitive radio networks \cite{CS00,CS01,CS03}. In such a non-cooperative setting with Rayleigh fading, the probabilities of miss (i.e. probability of missing the detection of the PU) and false alarm for an SU $i$ are,  respectively, given by $P_{m,i}$ and $P_{f,i}$ \cite[Eqs. (2), (4)]{CS00}\begin{align}\label{eq:missprob}
P_{m,i} =1- e^{-\frac{\lambda}{2}} \sum_{n=0}^{m-2}\frac{1}{n!}\left(\frac{\lambda}{2}\right)^n + \left(\frac{1+\bar{\gamma}_{i,PU}}{\bar{\gamma}_{i,PU}}\right)^{m-1}
\nonumber\\
\times \left[e^{-\frac{\lambda}{2\left(1+\bar{\gamma}_{i,PU}\right)}} - e^{-\frac{\lambda}{2}} \sum_{n=0}^{m-2}\frac{1}{n!}\left(\frac{\lambda\bar{\gamma}_{i,PU}}{2(1+\bar{\gamma}_{i,PU})}\right)^n \right],
\end{align}
\begin{equation}\label{eq:falseind}
P_{f,i} = P_f= \frac{\Gamma(m,\frac{\lambda}{2})}{\Gamma(m)},
\end{equation}
where $m$ is the time bandwidth product, $\lambda$ is the energy detection threshold assumed to be the same for all SUs without loss of generality as in \cite{CS00,CS01,CS03}, $\Gamma(\cdot,\cdot)$ is the incomplete gamma function and $\Gamma(\cdot)$ is the gamma function. Moreover, $\bar{\gamma}_i$ represents the average SNR of the received signal from the PU to SU given by $\bar{\gamma}_{i,\textrm{PU}} = \frac{P_{\textrm{PU}}h_{\textrm{PU},i}}{\sigma^2}$ with $P_{\textrm{PU}}$ the transmit power of the PU,  $\sigma^2$ the Gaussian noise variance, and $h_{\textrm{PU},i}=\kappa/d_{\textrm{PU},i}^{\mu}$ the path loss between the PU and SU $i$, $\kappa$ is the path loss constant, $\mu$ the path loss exponent, and $d_{\textrm{PU},i}$ is the distance between the PU and SU $i$.

The expression in (\ref{eq:missprob}) averages out the effect of the small-scale fading, and, thus, in the remainder of this paper, we only deal with large-scale fading effects, e.g., due to path loss. Further, it is also important to note that the non-cooperative false alarm probability expression depends \emph{solely} on the detection threshold $\lambda$ and does not depend on the SU's location; hence, we dropped the subscript $i$ in (\ref{eq:falseind}). Note that we do not account for the shadowing effect, however, in future work, (\ref{eq:missprob}) can be modified to take into account this effect as shown in \cite{CR02,CS00,CS01}.  Moreover, an important metric is the probability of detection of the PU by a SU $i$, which is simply defined as $P_{d,i} = 1 - P_{m,i}$.

In order to minimize the interference on the PU and reduce the probability of miss, the SUs can cooperate by forming coalitions. Within each coalition $S \subseteq \mathcal{N}$, an SU designated as \emph{coalition head}, collects the sensing bits from the coalition's SUs and acts as a fusion center in order to make a coalition-based decision on the presence or absence of the PU. This can be seen as having the centralized CSS of \cite{CS00}, \cite{CS03} applied at the level of each coalition with the coalition head being the fusion center to which all the coalition members report. For combining the sensing bits and making the final detection decision, the coalition head uses the decision fusion OR-rule such as in \cite{CS00}, \cite{CS03}. For a coalition $S$ having coalition head $k\in S$, using CSS, the miss and false alarm probabilities are, respectively, given by \cite{CS03}
\begin{align}\label{eq:missclu}
Q_{m,S} = \displaystyle \prod_{i \in S} \left[ P_{m,i}(1-P_{e,i,k}) + (1- P_{m,i})P_{e,i,k}\right],
\end{align}
\begin{equation}\label{eq:falseclu}
Q_{f,S} =  1 - \displaystyle\prod_{i \in S} \left[ (1-P_{f})(1-P_{e,i,k}) + P_{f}P_{e,i,k}\right],
\end{equation}
\noindent where $P_{m,i}$, $P_{f}$ are given by (\ref{eq:missprob}), (\ref{eq:falseind}), and $P_{e,i,k}$ is the probability of error due to the fading on the reporting channel between the SU $i$ of coalition $S$ and the coalition head $k$. The error over the reporting channel is an important metric that affects the performance of CSS in terms of probability of miss as well as false alarm as demonstrated in \cite{CS03}. The main idea is that, whenever an SU needs to report its sensing bit to the coalition head, the coalition head might receive this bit in error due to the fading over the wireless channel between the SU and the coalition head. This reporting error $P_{e,i,k}$, between an SU $i$ and the coalition head $k$, increases the probability of miss and increases the false alarm as shown in in (\ref{eq:missclu}) and (\ref{eq:falseclu}). Inside any coalition $S$, assuming BPSK modulation in Rayleigh fading environments, the average probability of reporting error between SU $i \in S$ and the coalition head $k \in S$ is given by \cite{PROAKIS}
 \begin{equation}\label{eq:err}
P_{e,i,k} = \frac{1}{2} \left(1-\sqrt{\frac{\bar{\gamma}_{i,k}}{1+ \bar{\gamma}_{i,k}}}\right),
\end{equation}
where $\bar{\gamma}_{i,k}=\frac{P_i h_{i,k}}{\sigma^2}$ is the average SNR for bit reporting between SU $i$ and the coalition head $k$ inside $S$ with $P_i$ the transmit power of SU $i$ used for reporting the sensing bit to $k$ and $h_{i,k}=\frac{\kappa}{d_{i,k}^{\mu}}$ the path loss between SU $i$ and coalition head $k$. Any SU can be chosen as a coalition head within a coalition. However, for the remainder of this paper, we adopt the following convention without loss of generality.
\begin{convention}
Within a coalition $S$, the SU $k \in S$ having the lowest non-cooperative probability of miss $P_{m,k}$ is chosen as \emph{coalition head}. Hence, the coalition head $k$ of a coalition $S$ is given by $k \in \underset{i\in S}{\operatorname{arg\,min}}\,P_{m,i}$ with $P_{m,i}$ given by (\ref{eq:missprob})\footnote{Note that, if more than one SU in a coalition achieves the minimum probability of miss, then, the coalition head is picked at random among the set of SUs with minimum probability of miss.}.
\end{convention}

The motivation behind Convention~1 is that, due to the error on the reporting channel,  whenever an SU needs to report its sensing bit to a coalition head, there is a potential ``risk'' that the bit is received in error as per (\ref{eq:missclu})-(\ref{eq:err}). This risk motivates the adoption of Convention~1 whereby the SU with minimum non-cooperative miss probability, i.e., the SU with the best channel towards the PU, is used as a coalition head, and, thus, does not need to send its bit (which is the most reliable in the coalition) to other SUs in the coalition and risk the reception of this bit  in error, subsequently, affecting the performance in terms of probability of miss and false alarm. It must be noted that the analysis and algorithm discussed in the remainder of this paper can accommodate other coalition head selection approaches (e.g., select the SU closest to the center of a coalition).

It should be clear from (\ref{eq:missclu}) and (\ref{eq:falseclu}) that as the number of SUs per coalition increases, the probability of miss will decrease while the probability of false alarm will increase. This is a crucial tradeoff in CSS that can have a major impact on the collaboration strategies of each SU. Thus, our objective is to characterize the network structure that will form
when the SUs collaborate while accounting for this tradeoff. An example of the sought network structure is shown in Fig.~\ref{fig:ill}. Finding the optimal coalition structure, such as in Fig.~\ref{fig:ill}, using a centralized approach requires finding the partition of the set that minimizes the average probability of miss per SU, while maintaining a false alarm constraint per SU.  For instance, in the centralized approach, the optimization problem aims at minimizing the overall system's average probability of miss per SU given a constraint on the average false alarm probability per SU. In other words, denoting $\mathcal{B}$ as the set of \emph{all partitions} of $\mathcal{N}$, the centralized approach seeks to solve the following optimization problem
 \begin{equation*}\label{eq:cent}
\min_{\mathcal{T} \in \mathcal{B}}{\frac{\sum_{S \in \mathcal{T}}|S|\cdot Q_{m,S}}{N} },\textrm{ s.t. }Q_{f,S}\le \alpha \ \forall\ S \in \mathcal{T},
 \end{equation*}
 where $|\cdot|$ represents the cardinality of a set operator and $S$ is a coalition belonging to the partition $\mathcal{T}$. Note that, in this centralized optimization problem it is considered that the false alarm probability constraint per SU directly maps to a false alarm probability constraint \emph{per coalition} which will be more formally demonstrated in the next section through Property~\ref{prop:equ}.

 However, it is shown in \cite{NP00} that finding the optimal coalition structure in a centralized manner leads to an optimization problem which is NP-complete. This is mainly due to the fact that the number of possible coalition structures (partitions), is given by a value known as the Bell number which grows exponentially with $N$ \cite{BN00}. Moreover, in a centralized partition search approach, the objective is typically to optimize the overall system performance while ignoring the individual preferences and payoffs of the SUs (or coalition of SUs). Hence, for the centralized approach, a particular SU might be asked to be part of a coalition $S_1$ even if this SU, based on its own payoff, prefers to be part of another coalition $S_2$. Since the SUs act independently and are typically self-interested, they should be given the choice of making their coalition formation decisions on their own. Hence, there is a motivation to devise distributed approaches for the SUs to make autonomous decisions, self-organize, and form the coalitional structure.
\begin{figure}[!t]
\begin{center}
\includegraphics[width=8cm]{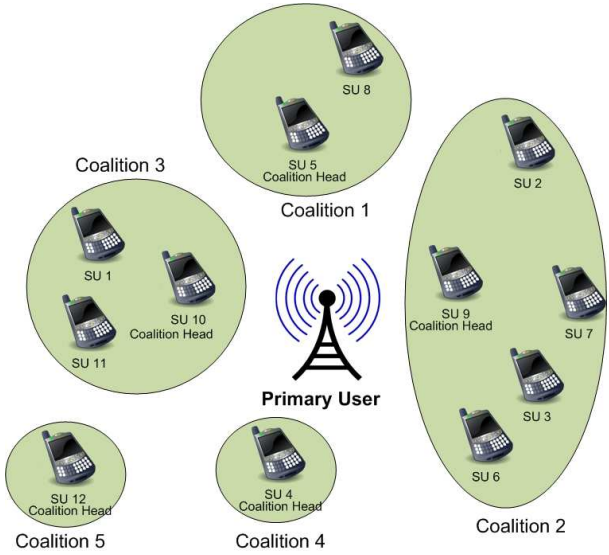}
\end{center}\vspace{-0.5cm}
\caption {An illustrative example of coalition formation for collaborative spectrum sensing among SUs.} \label{fig:ill}\vspace{-0.4cm}
\end{figure}

\section{Collaborative Spectrum Sensing As Coalitional Game}\label{sec:coal}\vspace{-0.1cm}
\subsection{Game Formulation and Properties}
For devising suitable cooperative strategies for CSS among the SUs, we refer to coalitional game theory \cite{Game_theory2,WS00,CF04}. We model distributed CSS as a coalitional game with a non-transferable utility which is defined as follows \cite{Game_theory2}\footnote{This is the general definition of an NTU game as per \cite{Game_theory2,WS00}, some literature may append other properties to the definition which are not useful for the game class dealt with in this paper (for more information on the various definitions of an NTU game see \cite{Game_theory2,WS00,CF04}).}
%\begin{definition}
%A coalitional game $(\mathcal{N},v)$ is said to have a \emph{transferable} utility if the value $v(S)$ can be arbitrarily apportioned between the coalition's players. Otherwise, the coalitional game %has a \emph{non-transferable utility} and each player will have its own utility within coalition $S$.
%\end{definition}
\begin{definition}
A coalitional game with non-transferable utility~(NTU) is defined by a pair $(\mathcal{N},V)$ where $\mathcal{N}$ is the set of players and $V$ is a mapping such that for every coalition $S \subseteq \mathcal{N}$, $V(S)$ is a closed convex subset of $\mathbb{R}^{S}$ that contains the payoff vectors that players in $S$ can achieve.
\end{definition}

In other words, a coalitional game $(\mathcal{N},V)$ is said to be with NTU if the  value or utility of a coalition cannot be arbitrarily apportioned between the coalition's players. Hence, the value of any coalition $S$ can be mapped to a set $V$ of payoff vectors.
First, we define the value (utility) $v(S)$ of a coalition $S \subseteq \mathcal{N}$ as a function that captures the tradeoff between the probability of detection and the probability of false alarm. For this purpose, $v(S)$ must be an increasing function of the detection probability $Q_{d,S} = 1 - Q_{m,S}$ within coalition $S$ and a decreasing function of the false alarm probability $Q_{f,S}$ as follows.
\setlength{\aligntop}{-1em}
\setlength{\alignbot}{-1.3\baselineskip}
\begin{align}\label{eq:util}
\!\!\!\!v(S)\! =\!  Q_{d,S}\!  - \! C(Q_{f,S},\alpha_S)\! =\!  (1 - Q_{m,S})\! -\! C(Q_{f,S},\alpha_S),
\end{align}\setlength{\aligntop}{-0.1em}\setlength{\alignbot}{-0.7\baselineskip}
\noindent where $Q_{m,S}$ is the probability of miss for coalition $S$ given by (\ref{eq:missclu}) and $C(Q_{f,S},\alpha_S)$ is a cost function which depends on the false alarm probability $Q_{f,S}$ and on a false alarm constraint $\alpha_S$ that $S$ must not exceed when it forms. Without loss of generality, we assume $\alpha_S = \alpha, \forall S \subseteq \mathcal{N}$.

In a coalitional game, one has to distinguish between two distinct quantities: The value function and the payoff. Given a coalition $S$, the value function $v(S)$ describes the overall utility that the \emph{entire} coalition $S$ receives when its acting cooperatively. In contrast, for any SU $i\in S$ the payoff $\phi_i(S)$ describes the utility that SU $i$ receives \emph{when acting as part of $S$}. Hence, $\phi_i(S)$ represents the tradeoff between probability of miss and probability of false alarm that SU $i$ achieves when acting as part of coalition $S$. In the proposed CSS game, we have the following property:

%Further, while $v(S)$ describes the utility of the coalition as a whole, we define, for every SU $i \in S$, $\phi_i(S)$ which is the payoff that SU $i$ receives whenever it is part of coalition $S$. Given this definition, we remark the following:
%In fact,  $\alpha$ maps directly to a false alarm costraint per SU as follows:
\begin{property}\label{prop:equ}
The probabilities of miss and false alarm for any SU $i \in S$ are given by the probabilities of miss and false alarm of coalition $S$ in (\ref{eq:missclu}) and (\ref{eq:falseclu}), respectively. Hence, in the proposed CSS game, the payoff, $\phi_i(S)$, of any SU $i\in S$ is simply equal to the value of the coalition, i.e., $\phi_i(S)=v(S), \ \forall\ i \in S$. As a result, the false alarm constraint per coalition $\alpha$ maps to a false alarm constraint per SU.
\end{property}
%In other words, for the proposed CSS game, it turns out that, the payoff of an SU $i \in S$ (which captures the performance of SU $i$) is simply equal to the value function, i.e., $\phi_i(S)=v(S)$.
 %In consequence, the utility of each SU $i \in S$ is, therefore, given by $\phi_i(S) = v(S)$.
%\begin{proof}
%Within each coalition $S$ the SUs report their sensing bits to the coalition head of $S$ which, after combining the results, makes a final decision on the presence or absence of the PU. Hence, SUs belonging to a coalition $S$ will transmit or not based on the final coalition head decision after data fusion. Therefore, the miss and false alarm probabilities of any SU $i \in S$ are the miss and false alarm probabilities of the coalition $S$ to which $i$ belongs as given by $Q_{m,S}$ and $Q_{f,S}$ in (\ref{eq:missclu}) and (\ref{eq:falseclu}), respectively. Thus, the utility of each SU $i \in S$ is given by $\phi_i(S) = v(S)$. As a result, the value of a coalition $v(S)$ cannot be arbitrarily apportioned among the SUs in a coalition; and the proposed coalitional game has NTU. Moreover, the false alarm constraint per coalition $\alpha$ maps to a false alarm constraint per SU.
%\end{proof}

This property is an immediate result of the operation of CSS whereby the SUs belonging to a coalition $S$ will transmit or not based on the final coalition head decision after data fusion. Thus, the miss and false alarm probabilities of any SU $i \in S$ are the miss and false alarm probabilities of the coalition $S$ to which $i$ belongs as given by $Q_{m,S}$ and $Q_{f,S}$ in (\ref{eq:missclu}) and (\ref{eq:falseclu}), respectively. As an immediate consequence of this property, we notice that, in the proposed CSS game, the value of a coalition, i.e., $v(S)$, is not divisible among the SUs since the payoff of \emph{every} SU in $S$ is equal to $v(S)$ (and not to a fraction of $v(S)$). Hence, the proposed CSS game can be modeled as a coalitional game $(\mathcal{N},V)$ with NTU where $V(S)$ is a \emph{singleton} set,
\begin{equation}\label{eq:utilmap}
V(S)=\{\boldsymbol{\phi}(S) \in \mathbb{R}^{S}|\ \phi_i(S) = v(S),\ \forall i \in S\},
\end{equation}
with $v(S)$ given by (\ref{eq:util}).  $(\mathcal{N},V)$ is clearly an NTU game since the set $V(S)$ is a singleton, and hence closed and convex.

In a coalitional game with \emph{no cost}, the main interest would be in characterizing the properties and stability of the grand coalition of all players since it is generally assumed that the grand coalition maximizes the utilities of the players \cite{Game_theory2} (in such a no cost case, concepts such as the core are suitable solutions and additional properties for $V$ may be required). However, for the proposed $(\mathcal{N},V)$ game, although CSS improves the detection probability for the SUs, the false alarm costs limit this gain, and thus, we remark the following:
 %When all players in $S$ receive a payoff larger than their non-cooperative payoff, i.e., if $\phi_i(S)=v(S) \ge v(\{i\}),\ \forall i\in S$, then $v(S)$ is surely finite  and hence the set of such vectors (i.e. the singleton vector $\boldsymbol{\phi}$) is bounded.
\begin{remark}
For the proposed $(\mathcal{N},V)$ coalitional game, the grand coalition of all the SUs \emph{seldom} forms due to the false alarm costs resulting from cooperation. Instead, disjoint independent coalitions form in the network.
\end{remark}

By inspecting $Q_{m,S}$ in (\ref{eq:missclu}) and through the results shown in \cite{CS00}, \cite{CS03} (for centralized CSS) it is clear that, as the number of SUs in a coalition increases, $Q_{m,S}$ decreases and the performance improves. Hence, when no cost for collaboration exists, the grand coalition of all SUs is the optimal structure for maximizing the detection probability. However, when the number of SUs in a coalition $S$ increases, through (\ref{eq:missclu}), the false alarm probability increases and can exceed the constraint $\alpha$, hence giving no incentive for large coalitions to form and affecting the collaboration strategies of the SUs. Therefore, for the proposed CSS model with cost for collaboration, the grand coalition of all SUs will seldom form due to the false alarm cost and constraint as accounted for in (\ref{eq:util}). Due to this property, we deal with a class of coalitional games known as coalition formation games \cite{WS00}. Hence, the solution for the proposed $(\mathcal{N},V)$ CSS coalition formation game is the network structure which can characterize the CSS trade off between gains (probability of miss) and costs (false alarm).%\vspace{-0.14cm}% Hence, we have a  coalition formation game and seek distributed algorithms for forming SUs coalitions.
\subsection{Cost Function}
Any well designed cost function $C(Q_{f,S},\alpha)$ in (\ref{eq:util}) must satisfy several requirements for adequately modeling the false alarm. On one hand, $C(Q_{f,S},\alpha)$ must be an increasing function of $Q_{f,S}$ with the slope becoming steeper as $Q_{f,S}$ increases. On the other hand, $C(Q_{f,S},\alpha)$ must impose a maximum tolerable false alarm probability $\alpha$ that cannot be exceeded by any SU (as per Property~1, imposing a false alarm constraint on the coalition maps to a constraint per SU). A well suited cost function satisfying these requirements is the logarithmic barrier penalty function given by \cite{BO00}
\begin{align}\label{eq:logbarr}
C(Q_{f,S},\alpha) = \begin{cases}- \alpha^2\cdot\log{\left(1-\left(\frac{Q_{f,S}}{\alpha}\right)^2\right)},& \mbox{if } Q_{f,S} < \alpha,\\ +\infty,&
\mbox{if } Q_{f,S} \ge \alpha, \end{cases}
\end{align}
where $\log$ is the natural logarithm and $\alpha$ is the false alarm constraint per coalition.  The cost function in (\ref{eq:logbarr}) incurs a penalty which is increasing with the false alarm probability. Moreover, it imposes a maximum false alarm probability $\alpha$ per coalition (per SU by Property~1). In addition, as the false alarm probability gets closer to $\alpha$, the collaboration cost increases steeply, requiring a significant improvement in detection probability if the SUs wish to collaborate as per (\ref{eq:util}). Note that the proposed cost function depends on both distance and the number of SUs in the coalition, through the false alarm probability $Q_{f,S}$ (the distance lies within the probability of error). Hence, the cost for collaboration increases with the number of SUs in the coalition as well as when the distance between the coalition's SUs increases.

\section{Coalition Formation Games in CSS}\label{sec:coalform}
\subsection{Distributed Coalition Formation Algorithm (CF)}\label{sec:mergeandsplit}
Coalition formation is a branch of game theory that investigates algorithms for studying the coalitional structures that form in a network where there is a cost for forming coalitions and where the grand coalition is not optimal \cite{WS00,NP00,KA01,CF04}. For constructing a distributed coalition formation algorithm for CSS, we use the following concepts \cite{WS00,KA01}.
\begin{definition}
A \emph{collection} $\mathcal{S}$ is the set $\mathcal{S} = \{S_{1},\ldots,S_{l}\}$ of mutually disjoint coalitions $S_{i} \subseteq \mathcal{N}$. If the collection spans \emph{all} the players of $ \mathcal{N}$, i.e., $\bigcup_{j=1}^{l} S_j =  \mathcal{N}$, the collection is a \emph{partition} of $\mathcal{N}$. %In other words, a collection is any arbitrary group of disjoint coalitions $S_i$ of $\mathcal{N}$ not necessarily spanning all players of $ \mathcal{N}$. I
\end{definition}
\begin{definition}
A \emph{comparison relation} $\rhd$ is defined for comparing two collections $\mathcal{R} = \{R_{1},\ldots,R_{l}\}$ and
$\mathcal{S} = \{S_{1},\ldots,S_{m}\}$ that are partitions of the same subset $\mathcal{A} \subseteq \mathcal{N}$ (same players in $\mathcal{R}$ and $\mathcal{S}$). Thus, $\mathcal{R} \rhd \mathcal{S}$ implies that the way $\mathcal{R}$ partitions $\mathcal{A}$ is preferred to the way $\mathcal{S}$ partitions $\mathcal{A}$ based on criteria defined next.
\end{definition}

Various criteria (referred to as \emph{orders}) can be used as comparison relations between collections or partitions \cite{WS00,KA01}. Due to the non-transferable nature of the proposed $(\mathcal{N},V)$ CSS game, an individual value order that compares the individual payoffs of the players must be used as a comparison relation $\rhd$. An adequate individual value order that can be used in CSS is the \emph{Pareto order} \cite{WS00,KA01}. Denote for a collection $\mathcal{R} = \{R_{1},\ldots,R_{l}\}$, the utility of a player $j$ in a coalition $R_j \in \mathcal{R}$ by $\phi_j(\mathcal{R})=\phi_j(R_j)=v(R_j)$ (Property~1); the Pareto order is defined as
\setlength{\aligntop}{-0.4em}
\setlength{\alignbot}{-1.6\baselineskip}
\begin{align}\label{eq:Pareto}
\mathcal{R} \rhd \mathcal{S} \Longleftrightarrow \{\phi_{j}(\mathcal{R}) \ge
\phi_{j}(\mathcal{S}) \  \forall \ j \in \mathcal{R},\mathcal{S}\},   \\ \textrm{ with \emph{at
least one strict inequality} ($>$) for a player } k. \nonumber
\end{align}
\setlength{\aligntop}{-0.1em}
\setlength{\alignbot}{-0.7\baselineskip}
Using the Pareto order as a comparison relation, we propose a distributed coalition formation algorithm based on two rules called ``merge'' and ``split'' that
allow to modify a partition $\mathcal{T}$ of the SUs set $\mathcal{N}$ as follows \cite{KA01}.
\begin{definition}
\textbf{Merge Rule -} Merge any set of coalitions
$\{S_{1},\ldots,S_{l}\}$ where $\{\cup_{j=1}^{l}S_{j}\} \rhd \{S_{1},\ldots,S_{l}\} $; therefore, {$\{S_{1},\ldots,S_{l}\}
\rightarrow \{\cup_{j=1}^{l}S_{j}\}$}.
\end{definition}
\begin{definition}
\textbf{Split Rule -} Split any coalition
$\cup_{j=1}^{l}S_{j}$ where  $\{S_{1},\ldots,S_{l}\} \rhd \{\cup_{j=1}^{l}S_{j}\} $; thus, $\{\cup_{j=1}^{l}S_{j}\}
\rightarrow \{S_{1},\ldots,S_{l}\}$.
\end{definition}

Using the above rules, multiple coalitions can interact, taking the decision to merge or split based on the comparison relation $\rhd$. For the Pareto order, coalitions decide to merge (split) only
if at least one SU is able to \emph{strictly} improve its individual utility through this merge (split) process without decreasing the other SUs' utilities. Therefore, the merge rule by Pareto order is a binding agreement among the SUs to operate together in the sensing phase if it is beneficial for them; and this agreement is partially reversible; i.e., can be reversed only by an agreement to split. Note that, coalition formation entails algorithms where agreements can be either irreversible, partially reversible, or fully reversible \cite{CF04}. As the algorithms become more reversible, their complexity and implementation can increase while their flexibility improves and this motivates using a partially reversible algorithm such as the one we propose which has a balance between implementation and flexibility \cite{CF04}. The use of these rules for CSS is described in details further in this section.
%A decision to merge or split is, thus, tied to the fact that all users must benefit from merge or split, thus, any merged (or split) form is reached only if it allows all involved users to maintain their utilities with at least one user improving.%  through socialist merge-and-split,  once a coalition is formed
%through merge, the users that agreed to participate in that coalition will commit to a binding agreement whereby no split will occur at a later time unless splitting
%improves the individual rate of all the users that formed the coalition.
%
%For instance, using the merge-and-split rules combined with the Pareto order, an autonomous and distributed coalition formation algorithm suited for CSS can be constructed. First and foremost, the appeal of forming coalitions using merge-and-split stems from the fact that it has been shown in \cite{KA01} that any arbitrary iteration of merge-and-split operations converges.  Moreover, each merge or split decision can be taken in a distributed manner by each individual SU or by each already formed coalition. Subsequently, a merge-and-split based algorithm can adequately model the distributed interactions among the SUs that are seeking to collaborate in the sensing process.

As a consequence, for the proposed CSS game, we build a distributed coalition formation~(CF) algorithm based on merge-and-split and it is divided into three phases: Discovery, adaptive coalition formation, and coalition sensing. In the first phase, each SU discovers neighboring SUs as well as information required for performing coalition formation, e.g., distance to the PU and the distance of neighboring SUs. The discovery phase is discussed in details in Subsection~\ref{sec:nd}. Following the discovery phase, adaptive coalition formation begins and the SUs (or existing coalitions of SUs) interact in order to assess whether to cooperate for sharing their sensing results with discovered coalitions. For this purpose, an iteration of sequential merge-and-split rules occurs in the network, whereby each coalition decides to merge (or split) depending on the utility improvement. In this phase, as time evolves and SUs (or the PU) move, the SUs can autonomously self-organize and adapt the network's structure through new merge-and-split iterations with each coalition taking the decision to merge (or split) subject to satisfying the merge (or split) rule through Pareto order (\ref{eq:Pareto}). In the final coalition sensing phase, once the network partition converges following merge-and-split, SUs that belong to the same coalition report their local sensing bits to their local coalition head. The coalition head subsequently uses decision fusion OR-rule \cite{CS00,CS03,CS04} to make a final decision on the presence or the absence of the PU. This decision is then reported by the coalition heads to all the SUs within their respective coalitions.  One round of the CF algorithm is summarized in Table~\ref{tab:alg1} and the detailed operation is as follows:
 % It must be noted that for complexity and overhead reduction, the local and coalition sensing phases may be repeated several times prior to adaptive coalition formation, namely in networks where environmental changes due to mobility are seldom.
 \begin{table}[!t]
%\scriptsize
  \centering
  \caption{%\mycaption{%\vspace*{-1em}
    \vspace*{-0.2em}One round of the proposed CF algorithm.}\vspace*{-1em}
    \begin{tabular}{p{8cm}}
      \hline
      \textbf{Initial State} \vspace*{.5em} \\
      \hspace*{1em} The network is partitioned by $\mathcal{T}=\{T_1,\ldots,T_k\}$ (At the beginning of all time $\mathcal{T}$ = $\mathcal{N}$ = $\{1,\ldots,N\}$ with non-cooperative SUs). \vspace*{.1em}\\
\textbf{Three phases in each round of the CF algorithm} \vspace*{.5em}\\
\hspace*{1em}\emph{Phase 1 - Discovery:}   \vspace*{.1em}\\
\hspace*{2em}Each SU (or coalition of SUs) discovers its neighboring\vspace*{.1em}\\
\hspace*{2em}coalitions as per Subsection~\ref{sec:nd}.\vspace*{.1em}\\
\hspace*{1em}\emph{Phase 2 - Adaptive Coalition Formation:}   \vspace*{.1em}\\
\hspace*{2em}In this phase, coalition formation using merge-and-split occurs. \vspace*{.2em}\\
\hspace*{3em}\textbf{repeat}\vspace*{.2em}\\
\hspace*{4em}a) $\mathcal{F}$ = Merge($\mathcal{T}$); coalitions in $\mathcal{T}$ decide to merge based\vspace*{.2em}\\
\hspace*{5em}on the merge algorithm explained in Subsection~\ref{sec:mergeandsplit}.\vspace*{.2em}\\
\hspace*{4em}b) $\mathcal{T}$ = Split($\mathcal{F}$); coalitions in $\mathcal{F}$ decide to split based on \hspace*{5em}the Pareto order.\vspace*{.2em}\\
\hspace*{3em}\textbf{until} merge-and-split converges.\vspace*{.2em}\\
\hspace*{1em}\emph{Phase 3 - Coalition Sensing:}   \vspace*{.1em}\\
\hspace*{2em}a) Each SU reports its sensing bit to the coalition head.\vspace*{.1em}\\
\hspace*{2em}b) The coalition head of each coalition makes a final decision on\vspace*{.1em}\\
\hspace*{2em}the absence or presence of the PU using decision fusion OR-rule.\vspace*{.1em}\\
\hspace*{2em}c) The SUs in a coalition abide by the final decision made by the\vspace*{.1em}\\
\hspace*{2em}coalition head.\vspace*{.1em}\\
\textbf{The above phases are repeated periodically throughout the network operation. In Phase~2, through distributed merge-and-split decisions, the SUs can autonomously adapt the network structure to environmental changes such as mobility. } \vspace*{.4em}\\
   \hline
    \end{tabular}\label{tab:alg1}\vspace{-0.3cm}
\end{table}

Each round of the three phases of the CF algorithm of Table~\ref{tab:alg1} starts from an initial network partition $\mathcal{T}_0 =\{T_1,\ldots,T_l\}$ of $\mathcal{N}$.  Following neighbor discovery each coalition $T_i \in \mathcal{T}_0$ can acquire the following information (see Subsection~\ref{sec:nd} for details):
\begin{itemize}
\item The probabilities of miss of the SUs in $T_i$, i.e., each SU in $T_i$ is aware of its own probability of miss.
\item The existence of other coalitions with whom a potential cooperation can be done. Hence, each coalition $T_i \in \mathcal{T}_0$ would have a list of all other potential coalitions with whom a cooperation is possible.
\item The distances between the SUs in $T_i$ and the SUs in the discovered coalitions.
\end{itemize}

Following the acquisition of this information through the discovery phase, the adaptive coalition formation phase of the algorithm begins and consists of the merge and split procedures. During the merge procedure, each SU (or coalition of SUs that would form afterwards) engages in pairwise negotiations (i.e., one-on-one) with discovered neighbors, over a control channel. One example of such a control channel is the ad hoc temporary control channel \cite{ADHOC00} that can be temporarily established by each SU (or coalition) for the purpose of negotiations (these temporary control channels are widely used in ad hoc networks \cite{ADHOC00} and, more recently, in cognitive networks  \cite{ND00,ND01,ZD00,ADHOC00}). By signalling over the control channels, two coalitions (or SUs) can inform each other of the identities of their members (e.g., the identity can be defined by the member's position in space) and, in the event of a potential merge, can decide to join together into a single coalition. For example, at the beginning of all time, when all the network is non-cooperative, following neighbor discovery, each SU $i \in \mathcal{N}$ can, over a control channel, start negotiating, in a pairwise manner with other discovered SUs. When negotiating with a discovered coalition or SU, SU $i$ can send its information and an intent to merge over the control channel. Subsequently, depending on the Pareto order as per (\ref{eq:Pareto}), the other coalition can respond with its own information and an approval or denial of the merge.

Essentially, the \emph{merge process} occurs as follows: Given partition $\mathcal{T}_0$, each coalition $T_i \in \mathcal{T}_0$ negotiates over a control channel (e.g., through its coalition head) with the list of neighbors for assessing a potential merge. Proceeding sequentially through the list of discovered neighbors, each coalition head $k\in T_i$, attempts to merge with one of discovered neighbors $T_{\textrm{disc}}$ (at the beginning of all time all the discovered neighbors are non-cooperative SUs) as follows (the sequence of this procedure is assumed to be arbitrary depending on which neighbor was first discovered):

\begin{enumerate}

\item Coalition head $k \in T_i,\ \forall T_i \in \mathcal{T}_0$ signals an intent to merge with $T_{\textrm{disc}}$ by  sending, over the control channel, the information (probabilities of miss and distance estimates) on its own coalition, i.e., $T_i$, to one of the SUs in the coalition $T_{\textrm{disc}}$.
\item The SU in $T_{\textrm{disc}}$ which receives this information, computes, in cooperation with the members of its coalition $T_{disc}$, the potential probability of miss that would result from the merge with $T_i$ as per (\ref{eq:missclu}), and, following this computation, coalition $T_{\textrm{disc}}$ would respond to $k$ with one of the following results:
    \begin{itemize}
    \item If the Pareto order as per (\ref{eq:Pareto}) is not verified, i.e., the merger would decrease the payoff of one or more of the coalition members in $T_i$ or $T_{disc}$, then $T_{\textrm{disc}}$ responds with a small packet signalling the failure of the merger as well as providing information on its current members. The information on the current members allows $k$ and $T_i$ to get a better idea of the existing partition for future merge decisions.
            \item If the Pareto order as per (\ref{eq:Pareto}) is verified, i.e., cooperation improves the payoffs of the involved SUs, then $T_{\textrm{disc}}$ responds with a small packet signalling the acceptance of the merger as well as providing information on its current members. Subsequently, the two coalitions would merge into a new coalition $T_{\textrm{new}} \cup T_{\textrm{disc}}$ and they would be aware of each other following the negotiations phase. Subsequently, the new coalition $T_{\textrm{new}}$ would repeat the above procedure, given the union of the lists of neighbors from $T_i$ and $T_{\textrm{disc}}$.
    \end{itemize}

    \item Coalition $T_i$ continues the above procedure until it merges or until no further merge possibility exist.
\end{enumerate}
Hence,  two coalitions would merge through negotiations over a control channel, and they can inform each others of the identities of their members through this channel when the merge procedure above is taking place. Note that, although we consider that the order in which the SUs discover their neighbors (and goes through them for merge) is arbitrary, one can always define other metrics such as the average distance between the members which can be used as a distance between coalitions and, subsequently, used to order the discovered neighbors. In many cases, from the perspective of a coalition $S$, any neighbor discovery algorithm is likely to find coalitions whose members are at a closer distance from the members of $S$, before those who are located far away.

Moreover, the aforementioned procedure shows that, in general, the SUs do not need to estimate the probability of miss of their neighbors. In contrast, during the pairwise merge search, each SU (or coalition) would append the values of the probabilities of miss of its members along with the intent to merge message that is sent to one of its neighbors. In response, the neighbor can communicate, along with the merge decision, the information on the probability of miss of its members\footnote{In case the SUs, using the techniques of Subsection~\ref{sec:nd}, are able to estimate the probability of miss of their neighbors, they may be able to reduce the information exchange overhead.}.  Certainly, we assume that each SU is aware of its \emph{own} probability of miss and the methods that each SU can use to find its \emph{own} probability of miss are discussed in details in Subsection~\ref{sec:nd}. Note that, once the SUs of two coalitions exchange their probabilities of miss, they would be able to work out which member is potentially the coalition head as per Convention~1.

Following the merge process, the coalitions are next subject to split operations, if any is possible. Hence, any coalition which can no longer merge would start to apply the split operation. As the members are part of the same coalition, assessing the payoffs yielded by a split form is easily performed since the members are aware of the identities and characteristics of their partners. An iteration consisting of multiple successive merge-and-split operations is repeated until it terminates (the convergence of an algorithm constructed using merge-and-split rules is guaranteed \cite{KA01}). It must be stressed that the merge or split decisions are taken in a distributed way without relying on any centralized entity as each SU or coalition makes its own merge or split decisions.

The merge-and-split procedure of Phase~2 is performed sequentially for all coalitions in the network, in a certain order (in practice there is always a fraction of time of difference between the action of a coalition/SU or the next one, which dictates an order of operation). The order in which the coalitions/SUs act is considered arbitrary which is often the case in a practical cognitive network. As a result, any arbitrary SU (or coalition) can initiate the merge and split procedures. Hence, once the SUs have operated non-cooperatively for a while and performed the  discovery phase, they would start, in an arbitrary order, to perform coalition formation in Phase~2 through merge-and-split. In general, all SUs will always have an incentive to seek cooperation since it might improve their payoff. Nonetheless, the proposed algorithm is always applicable regardless of the order in which the SUs/coalitions act, and, thus, a network operator can impose a certain order for coalition formation on its SUs if needed (e.g., let the SU with worst or best performance start and so on).

Further, it is interesting to note that, under certain conditions (see Subsection~\ref{sec:stab} and Lemma~1) on the network partition, the order of operation does not affect the final structure that will form in the network. Nonetheless, even when different orders can yield different networks structures, due to the definition of the merge and split operation through the Pareto order in (\ref{eq:Pareto}), it is always ensured that, any resulting partition would improve the payoffs of the SUs as compared to the previous state. This is due to the fact that any merge or split agreement makes sure that at least one SU benefits from this operation. Hence, the proposed merge-and-split algorithm, independent of the order, is guaranteed to yield a cooperative network which outperforms (or is at least as good as) the non-cooperative approach.

For handling environmental changes such as mobility (of the SUs or the PU) or the joining/leaving of SUs, Phase~2 of the proposed algorithm in Table~\ref{tab:alg1} is repeated periodically over time. In Phase~2, periodically, as time evolves and SUs (or the PU) move or join/leave,  the SUs can take distributed decisions to adapt the network's structure through new merge-and-split iterations with each coalition taking the decision to merge (or split) subject to satisfying the merge (or split) rule through the Pareto order in (\ref{eq:Pareto}).

For this purpose, one can define a period of time $\theta$ which specifies the time elapsed between two consecutive runs of the proposed CF algorithm.  Consequently, there would be a tradeoff between the number of runs, i.e., overhead for coalition formation and the adaptation to the dynamics of the environment. Hence, for environments that are static or varying very slowly, $\theta$ would be large. In contrast, for highly varying environments, $\theta$ would have a small value, and, thus, enabling adaptation to rapidly changing environments. The value of $\theta$ can be regulated by the SUs depending on their observations on the environment. For example, given an initial value of $\theta=\theta_0$, whenever the SUs notice that, during merge-and-split, a lot of changes in the network and the partition have occurred, they can decide to reduce $\theta$ to a value smaller than $\theta_0$. In contrast, whenever the SUs note that few changes have occurred, they can increase $\theta$ to a value above $\theta_0$. For example, the possible values for the period $\theta$ can be standardized so that the SUs would know it without the need of a centralized controller. The above procedure would ensure the convergence of the merge-and-split algorithm. Nonetheless, certainly, we consider only environments where the changes during a single run of the proposed algorithms are small enough so as not to drastically affect the utilities, and, thus, the merge or split decisions of the SUs. Note that, performing instantaneous coalition formation is still a challenging task in game theory (e.g., see \cite{CF04} and references therein), and can be considered in future extensions of this work.

% In other words, every period of time $\theta$ the SUs assess the possibility of splitting into smaller coalitions or merging with new partners. The period $\theta$ is smaller in highly mobile environments to allow a more adequate adaptation of the structure. Similarly, every period $\theta$, in the event where the current coalition head of a coalition has moved or is turned off, the coalition members may select a new coalition head if needed. The convergence of this merge-and-split adaptation to environmental changes is always guaranteed as previously explained in this section.

It must be also noted that almost no synchronization among the SUs is needed to perform coalition formation since the order in which the merge-and-split operations proceed can be arbitrary (or synchronized if an operator wishes to do so). In some sense, the SUs would need only some kind of quasi-synchronization for handshaking and negotiations over the control channel to communicate their spectral sensing information. The proposed merge-and-split CF process involves negotiations that occur in a pairwise manner with coalitions (or SUs) in a small neighboring area and not in the entire network. In addition, as soon as a coalition identifies a merge possibility, it does not need to go through the remainder of the discovered neighbor. This aspect further corroborates the little need for synchronization among the SUs.

In order to show that the SUs need to survey only a specific geographical area for finding merge partners without looking at all other possible SUs in the network, we inspect the different merge possibilities between two SUs. In this regards, given two SUs involved in CSS through CF, we have the following theorem (proof is in the appendix):
\begin{theorem}\label{th:one}
 Two SUs $i$ and $j$ with non-cooperative probabilities of miss $P_{m,i}$ and $P_{m,j}$ respectively, such as $P_{m,i} \le P_{m,j}$, can merge into one coalition if the probability of error is below a
maximum probability of error, i.e., \emph{if} $P_{e,j,i} \le \tilde{P}_{e,j,i}$, where $\tilde{P}_{e,j,i}$ is given by the following tight approximation
\begin{align}\label{eq:th}
\tilde{P}_{e,j,i} \approx \frac{ P_f(P_f -2) + \alpha \sqrt{(1-(1-\frac{P_f^2}{\alpha^2})e^{\frac{P_{m,i}(1-\eta)(P_{m,j}-1) }{\alpha^2}} )}}{(2P_f -1)(P_f - 1)},
\end{align}
with $ \eta = \alpha$, if $P_{m,j} \le \frac{1}{2}$ and $\eta= 0$ otherwise. Consequently, a corresponding minimum distance $\tilde{d}_{i,j}$ between the two SUs can be easily found by solving (\ref{eq:err}) for $\tilde{P}_{e,j,i}$.
\end{theorem}

From Theorem~\ref{th:one}, each SU $i$ can determine merge regions, corresponding to arcs on circles and their respective intercepting angles, for forming a coalition with an SU $j$ where $P_{m,i} \le P_{m,j}$. For a given detection threshold $\lambda$, one can see through (\ref{eq:missprob}) that each $P_{m,j}$ corresponds to \emph{one circle} centered at the PU with a radius $r_j$ hence determining a class of SUs having the same non-cooperative probability of miss. For each such circle, an SU $i$ maps the condition  $P_{e,j,i} \le \tilde{P}_{e,j,i}$ to an arc on these circles, with a corresponding intercepting angle (angle centered at SU $i$) and intercepting the merge region arc. However, for any two SUs $i$ and $j$ belonging to two circles of radii $r_i$ and $r_j$, the distance $\tilde{d}_{i,j}$ between the two SUs (from Theorem~1) must satisfy $\tilde{d}_{i,j} < (r_j - r_i)$, otherwise no merge is possible (the merge region is empty). A detailed analysis of these merge regions, using Theorem~\ref{th:one}, is presented in Section~\ref{sec:sim}. Note that, even if no closed form expression exists for the upper bound $\tilde{P}_{e,j,i}$, $P_{e,j,i}$ will always be bounded by a value $\tilde{P}_{e,j,i}$ which would be a function of the probabilities of miss and false alarm of the two SUs. This is a direct consequence of the fact that, whenever the error on the reporting channel increases, it will incur an increased probability of miss and an increased false alarm as per (3) and (4) which, in turn, would limit the gains from cooperation and sets an upper bound on the error level that would be tolerated for the merge to occur.

\indent An upper bound on the maximum coalition size is imposed by the proposed utility and  models as follows:
\begin{theorem}
For the proposed CSS model, given that the SUs collaborate for improving their payoff given the tradeoff between probability of miss and false alarm, any coalitional structure resulting from the distributed coalition formation algorithm will have coalitions limited in size to a maximum of
\begin{align}
M_{\max}=\frac{\log{(1-\alpha)}}{\log{(1-P_f)}}
\end{align}
SUs.
\end{theorem}
\begin{proof}
In the proposed algorithm, the benefit from collaboration is limited by the false alarm probability cost modeled by the barrier function (\ref{eq:logbarr}). A minimum false alarm cost in a coalition  $S$ with coalition head $k\in S$ exists whenever the reporting channel is perfect, i.e., exhibiting no error, hence, $P_{e,i,k} = 0\ \forall i\in S$. In this perfect case, the false alarm probability in a perfect coalition $S_p$ is given by
\setlength{\aligntop}{-0.6em}
\setlength{\alignbot}{-1\baselineskip}
\begin{align}\label{eq:perf}
Q_{f,S_p} =  1 - \displaystyle\prod_{i \in S_p}  (1-P_{f}) = 1- (1-P_f)^{|S_p|},
\end{align}
\setlength{\aligntop}{-0.1em}
\setlength{\alignbot}{-0.7\baselineskip}
\noindent where $|S_p|$ is the number of SUs in the perfect coalition $S_p$. A perfect coalition $S_p$ where the reporting channels inside are perfect (i.e., SUs are grouped very close to each other) accommodates the largest number of SUs relative to all other coalitions. Hence, we use this perfect coalition to find an upper bound on the maximum number of SUs per coalition. For instance, the log barrier function in (\ref{eq:logbarr}) tends to infinity whenever the false alarm  constraint per coalition is reached which implies an upper bound on the maximum number of SUs per coalition if
%\begin{equation}
$Q_{f,S_p}\ge \alpha,$
%\end{equation}
yielding by (\ref{eq:perf}), $|S_p| \le \frac{\log{(1-\alpha)}}{\log{(1-P_f)}} = M_{\max}$.
\end{proof}

It is interesting to note that the maximum size of a coalition $M_{\max}$ depends mainly on two parameters: the false alarm constraint $\alpha$ and the non-cooperative false alarm $P_f$. For instance, larger false alarm constraints allow larger coalitions, as the maximum tolerable cost limit for collaboration is increased. Moreover, as the non-cooperative false alarm $P_f$ decreases, the possibilities for collaboration are better since the increase of the false alarm due to coalition size becomes smaller as per (\ref{eq:falseclu}).  It must be noted that the dependence of $M_{\max}$ on $P_f$ yields a direct dependence of $M_{\max}$ on the energy detection threshold $\lambda$ as per (\ref{eq:falseind}). Finally, it is interesting to see that the upper bound on the coalition size does not depend on the location of the SUs in the network nor on the actual number of SUs in the network. Therefore, deploying more SUs or moving the SUs in the network for a fixed $\alpha$ and $P_f$ does not increase the upper bound on coalition size.

\mysubsection{Partition Stability}\label{sec:stab}
The result of the proposed algorithm is a
network partition composed of disjoint independent coalitions of SUs. The stability of this resulting network structure can be investigated by a defection function $\mathbb{D}$ \cite{KA01}.
\begin{definition}
A \emph{defection} function $\mathbb{D}$ is a function which
associates with each partition $\mathcal{T}$ of $\mathcal{N}$ a group of
collections in $\mathcal{N}$. A partition $\mathcal{T} = \{T_1,\ldots,T_l\}$ of $\mathcal{N}$ is \emph{$\mathbb{D}$-stable} if no group of players is interested in leaving $\mathcal{T}$ when the players who leave can only form the collections allowed by $\mathbb{D}$.
\end{definition}

Two defection functions are of interest: a $\mathbb{D}_{hp}$ function which associates with each partition $\mathcal{T}$ of $\mathcal{N}$ the group of all partitions of $N$ that can form
by merging or splitting coalitions in $\mathcal{T}$ and the $\mathbb{D}_{c}$ function
which associates with each partition $\mathcal{T}$ of $\mathcal{N}$ the group of all collections in $\mathcal{N}$. Two forms of stability stem from these definitions: $\mathbb{D}_{hp}$ stability and a stronger $\mathbb{D}_{c}$ stability. A partition $\mathcal{T}$ is $\mathbb{D}_{hp}$-stable, if, for the partition $\mathcal{T}$, no coalition has an incentive to split or merge. As an immediate result of the definition of $\mathbb{D}_{hp}$-stability we have that every partition resulting from our proposed CF algorithm is \emph{$\mathbb{D}_{hp}$-stable}. Briefly, a $\mathbb{D}_{hp}$-stable can be thought of as a state of equilibrium where no coalitions have an incentive to pursue coalition formation through merge or split.

%\begin{lemma}
%\end{lemma}

%\begin{proof}
%Assume  $\mathcal{T}=\{T_1,\ldots,T_l\}$ is the
%partition resulting from our algorithm in Table~\ref{tab:alg1}.
%If for any $i\in \{1,\ldots,l\}$ and for any partition $
%\{R_1,\ldots,R_m\}$ of $T_i \in \mathcal{T}$, we assume that
%$\{R_1,\ldots,R_m\}\rhd T_i$ then the partition $\mathcal{T}$ can still be
%modified by applying a \emph{split} rule on $T_i$
%contradicting with the fact that $\mathcal{T}$ resulted from a termination
%of the merge-and-split iteration; therefore,
%$\{S_1,\ldots,S_m\}\ntriangleright T_i$ (first $\mathbb{D}_{hp}$
%stability condition verified). A similar reasoning is applicable
%in order to prove that $\mathcal{T}$ verifies the second condition, since
%otherwise a merge rule would still be applicable.
%\end{proof}
With regards to $\mathbb{D}_{c}$ stability,  a $\mathbb{D}_{c}$-stable partition has the following properties \cite{KA01}: (i)- In a $\mathbb{D}_{c}$-stable partition $\mathcal{T}$, no players are interested in leaving $\mathcal{T}$ to form other collections in $\mathcal{N}$ (through any operation), (ii)- If it exists, a $\mathbb{D}_{c}$-stable partition is the
\emph{unique} outcome of any \emph{arbitrary} iteration of merge-and-split and is a $\mathbb{D}_{hp}$-stable partition, and, (iii)- A $\mathbb{D}_{c}$-stable partition $\mathcal{T}$ is a
unique $\rhd$-maximal partition, that is for all partitions $ \mathcal{T}'
\neq \mathcal{T}$ of $\mathcal{N}$, $\mathcal{T} \rhd \mathcal{T}'$. When $\rhd$ is
the Pareto order, the $\mathbb{D}_{c}$-stable
partition presents a \emph{Pareto
optimal} utility distribution.

%It must be noted that the concept of $\mathbb{D}_{c}$-stability is referred to as strict $\mathbb{D}_{c}$-stability in \cite{KA00}  and as $\mathbb{D}_{c}$-stability
% in \cite{KA01}.
%Clearly, a $\mathbb{D}_{c}$-stable partition is an optimal
%partition which provides a
%utility distribution that is Pareto optimal for all SUs with
%respect to any other network partition.

However, the existence of a $\mathbb{D}_{c}$-stable partition is not always guaranteed \cite{KA01}. The  $\mathbb{D}_{c}$-stable partition $\mathcal{T} = \{T_{1},\ldots,T_{l} \}$ of the whole space $\mathcal{N}$ exists if a partition of $\mathcal{N}$ that verifies the following two necessary and sufficient conditions exists\cite{KA01}:
 \begin{enumerate}[(i)]
  \item For each $i\in \{1,\ldots,l\}$ and each pair of disjoint \emph{coalitions}  $S_1$ and $S_2$ such that $\{S_1 \cup S_2\} \subseteq T_i$, we have $\{S_1 \cup S_2\} \rhd \{S_1,S_2\}$.
       \item For the partition $\mathcal{T}=\{T_1,\ldots,T_l\}$ a coalition $G \subset \mathcal{N}$ formed of players belonging to different $T_i \in \mathcal{T}$ is $\mathcal{T}$-incompatible if for no
$i \in \{1,\ldots,l\}$ we have $G\subset T_i$. $\mathbb{D}_{c}$-stability requires that for  all $\mathcal{T}$-incompatible
coalitions $\{G\}[\mathcal{T}] \rhd \{G\}$ where $\{G\}[\mathcal{T}] = \{G\cap T_i \ \forall \ i\in\{1,\ldots,l\}\}$ is the projection of coalition $G$ on $\mathcal{T}$.
\end{enumerate}
If no partition satisfies these conditions, then no $\mathbb{D}_{c}$-stable partitions of $\mathcal{N}$ exist. Nevertheless, we have
\begin{lemma}
For the proposed $(\mathcal{N},V)$ CSS coalitional game, the proposed CF algorithm converges to the optimal $\mathbb{D}_{c}$-stable partition, if such a partition exists. Otherwise, the final network partition is $\mathbb{D}_{hp}$-stable.
\end{lemma}
\begin{proof}
The proof is immediate due to the fact that, when it exists, the  $\mathbb{D}_{c}$-stable partition is a unique outcome of any merge-and-split iteration \cite{KA01} such as any partition resulting from our CF algorithm.
\end{proof}

For the proposed game, the existence of the $\mathbb{D}_{c}$-stable partition cannot always be guaranteed although the algorithm would \emph{always} converge to this partition when it exists. Given a partition $\mathcal{T}=\{T_1,\ldots,T_l\}$ of SUs, verifying condition (i) for $\mathbb{D}_c$-stability implies that, for any coalition $T_i \in \mathcal{T}$, the Pareto order as per (\ref{eq:Pareto}) must be verified for the \emph{union} of every pair of sub-coalitions of $T_i$ which, in the context of collaborative sensing, maps directly into a condition dependent on the distances between the SUs inside $T_i$ and their non-cooperative probabilities of miss and false alarm (due to the utility in (\ref{eq:util})). For example, if $T_i$ is a coalition of size $2$, condition (i) implies that Theorem~1 and (\ref{eq:th}) must be verified between the two SUs of $T_i$. Further, if $T_i$ is a coalition of size $3$, condition (i) implies the following:
\begin{enumerate}[(a)]
\item  The distance between every pair of SUs in $T_i$ must satisfy (\ref{eq:th}) and Theorem~1.
 \item The distances and probabilities of miss and of false alarm of the SUs in $T_i$ must satisfy the merge rule  using the Pareto order in (\ref{eq:Pareto}) (in any direction).
 \end{enumerate}

Consequently, simple examples of partitions satisfying condition (i) of  $\mathbb{D}_c$-stability would be a partition composed solely of coalitions of size $2$ which verifies (\ref{eq:th}) and Theorem~1, a  partition $\mathcal{T}$ composed of coalitions of size $3$ that verifies (a) and (b), or a partition composed of a mixture of coalitions of size $2$ or $3$ satisfying the requirements previously stated. Certainly, for larger coalitions, condition (i) imposes restrictions on the distances and probabilities of miss and false alarm of the cooperating SUs analogous to (\ref{eq:th}), (a), and (b).

The difficulty of guaranteeing a $\mathbb{D}_c$-stable partition in the context of collaborative sensing stems from the difficulty of computing a closed form solution that characterizes condition (ii). For a partition $\mathcal{T}$, condition (ii) deals with $\mathcal{T}$-incompatible coalitions, which are coalitions composed of players belonging to different coalitions $T_i \in \mathcal{T}$. In order to clarify this requirement, consider a $3$-player game and a partition $\mathcal{T} = \{\{1,2\},\{3\}\}$. For this case, the $\mathcal{T}$-incompatible coalitions are $G_1=\{2,3\},\ G_2=\{1,3\},\ G_3=\{1,2,3\}$. For every $G_i$, condition (ii) requires that $\{G_i \cap \{1,2\}, G_i\cap \{3\}\} \rhd \{G\}$ which maps into:
\begin{enumerate}[(A)]
\item $\{\{2\},\{3\}\} \rhd \{2,3\}$ which implies that SUs $2$ and $3$ must \emph{not} verify (\ref{eq:th}) as per Theorem~1.
\item $\{\{1\},\{3\}\} \rhd \{1,3\}$ which implies that SUs $1$ and $3$ must \emph{not} verify (\ref{eq:th}) as per Theorem~1.
\item $\{\{1,2\},\{3\}\} \rhd \{1,2,3\}$ which implies that the Pareto order must not be satisfied, i.e., forming coalition $\{1,2,3\}$ decreases the payoff of at least one of the SUs.
\end{enumerate}

Consequently, for the $3$-player example, verifying condition (ii) for the existence of a $\mathbb{D}_c$-stable partition depends, mainly, on the distances between the $3$ SUs and their location with respect to the PU (through the probability of miss) as per Theorem~1. For larger networks, satisfying condition (ii) would, once again, depend on the distances between the SUs of $\mathcal{T}$-incompatible coalitions, however, the number of such coalitions becomes large, and, thus, mapping the requirement into closed form conditions such as in Theorem~1 would be difficult. Note that, for the $3$-player example previously mentioned, in the event where (A)-(C) are satisfied, and, if $\{1,2\}$ satisfies (\ref{eq:th}), then, partition $\mathcal{T} = \{\{1,2\},\{3\}\}$ would verify conditions (i) and (ii), and, thus, it will be a $\mathbb{D}_c$-stable partition and a unique outcome of the proposed algorithm.

In summary, for verifying condition (i) for existence of the  $\mathbb{D}_{c}$-stable partition, the SUs belonging to partitions of each formed coalition must verify the Pareto order as per (\ref{eq:Pareto}) using their utility defined in (\ref{eq:util}). Similarly, for verifying condition (ii) of $\mathbb{D}_{c}$ stability,  SUs belonging to all $\mathcal{T}$-incompatible coalitions in the network must verify the Pareto order. As shown in the previous examples, as well as through Theorem~1 (for $2$ SUs), this existence is closely tied to the location of the SUs and the PU (due to the dependence on the individual miss and false alarm probabilities in the utility expression (6)) which both can be random parameters in practical and large networks. Nonetheless, the proposed algorithm will always guarantee convergence to this optimal  $\mathbb{D}_{c}$-stable partition when it exists as stated in Lemma~1. Whenever a $\mathbb{D}_{c}$-stable partition does not exist, the coalition structure resulting from the proposed algorithm will be $\mathbb{D}_{hp}$-stable as no coalition or SU is able to merge or split any further (can be viewed as a kind of partition \emph{equilibrium}).

\section{Distributed Coalition Formation Algorithm with Probability of Detection Guarantees  (CF-PD)}\label{sec:dcfpd}
The previously proposed CF algorithm in Section~\ref{sec:coalform} allows the SUs in a cognitive network to minimize their probability of miss under false alarm constraints. In such a scenario, one can see the PU rewarding the SUs as much as they improve their detection and reduce their interference level, hence, providing an incentive for the SUs to lower their probability of miss as much as possible. However, in some cognitive networks, in addition to the false alarm constraint that the SUs require, the PU (or PU operator) imposes a desired detection probability $\chi$ that the SUs must meet (i.e., a desired interference level) \cite{CR02}. The PU introduces such constraints in order to control the amount of interference that is allowed by SUs or to limit and control the number of SUs that are allowed to use the spectrum. The presence of such an additional constraint can affect the behavior of the SUs, their incentives, and their collaboration strategies. Hence, the coalition formation process must be adequately modified in order to cope with the presence of such an additional detection probability constraint as the objective of the SUs is no longer to minimize their probability of miss, but rather to achieve the desired probability of miss $\gamma=1-\chi$ (while meeting the false alarm constraint). In this context, during a merge or split process, once a coalition achieves the probability of miss $\gamma$, it will have no further incentive to pursue the coalition formation process. In fact, in the presence of detection probability constraints, once a coalition achieves $\gamma$ increasing the size of the coalition incurs an increased false alarm (within the constraint $\alpha$), as well as extra signalling and sensing time with no extra benefit for the coalition members.

In order to appropriately capture these new objectives in the coalition formation process, we introduce, adjunctly with the proposed $(\mathcal{N},V)$ coalitional game, a \emph{simple voting game}  \cite{Game_theory1} which is defined as a coalitional game with utility function $u: 2^{\mathcal{N}} \rightarrow \{0,1\}$.  By introducing this simple voting game in the coalition formation process, the new objectives of the SUs, i.e., probability of detection guarantees, can be accommodated as will be seen later in this section. For this purpose, we define the \emph{adjunct utility function} $u$ as follows:
\begin{align}\label{eq:utilv}
u(S) = \begin{cases} 1, & \mbox{if }  Q_{d,S}  \ge \chi \textrm{ and } Q_{f,S} \le \alpha, \\ 0, &
\mbox{otherwise}. \end{cases}
\end{align}

This definition captures the two objectives of the SUs: (i) Satisfying the probability of detection $\chi$, and, (ii) Maintaining a false alarm level below $\alpha$, i.e., keeping $Q_{f,S} \le \alpha$. Thus, in (\ref{eq:utilv}), $u(S)=1$ if the probability of detection satisfies $\chi$ and the false alarm constraint is \emph{not} violated. For a coalition $S$ where $Q_{f,S} > \alpha$, as per (\ref{eq:util}), the utility goes towards negative infinity, and, thus a coalition with $Q_{f,S} > \alpha$ can never form, hence, in this case $u(S)=0$. During coalition formation, as will be seen later in this section, the SUs can utilize both utilities $V$ and $u$ in (\ref{eq:utilmap}) and (\ref{eq:utilv}) in order to make their cooperative decisions given the objective of achieving the target probability of detection while making sure that doing so benefits them as per (\ref{eq:utilmap}) while the false alarm constraint is not violated. Further, within the proposed adjunct simple voting game framework, we make the following definition:
 \begin{definition}
In the proposed CSS game, a coalition $S$ is said to be \emph{winning} if $u(S) = 1$ and losing otherwise. Similarly, a player is winning if it is part of a winning coalition and losing otherwise.
\end{definition}

Consequently, under a desired probability of detection constraint, the objective is to maximize the number of winning players while minimizing the cost in terms of false alarm (given the false alarm constraint $\alpha$) as well as the signaling overhead within each formed coalitions. In this regard, the partition stability concepts introduced in Section~\ref{sec:stab} are no longer adequate to assess the performance of the system since the SUs are no longer looking to maximize their utilities but rather to form winning coalitions. Therefore, a more adequate notion for stability is needed through the following concept.

\begin{definition}
In a simple voting game $(\mathcal{N},u)$, a coalition $S$ is said to be \emph{minimal winning} if the defection of \emph{any} player in $S$ renders that coalition losing. In other words, a \emph{minimal winning coalition (MWC)} $S$ is a coalition such that $u(S)=1\ \&\ u(T) = 0, \forall \ T \subset S$.
\end{definition}
From this definition, it follows that, in order to achieve the detection probability $\chi$ with minimum overhead, the SUs have an incentive to form MWCs. Using Definition~8 and (\ref{eq:utilv}), one can see that, whenever a group of SUs form an MWC, they have no incentive to leave this MWC. In fact, when an SU leaves an MWC, as per Definition~8, this SU becomes losing, i.e., its adjunct utility $u$ becomes $0$. Thus, an MWC is highly suitable to characterize the coalitional structure that forms when \emph{both} false alarm and detection probability constraints are required. Moreover, once the SUs form an MWC, they have no further incentive to decrease their probability of miss (which can incur extra false alarm) as they already comply with the requirements set by the PU.

Jointly with the merge and split operations, we define a new operation on the coalitions as follows:
\begin{definition}
\textbf{Adjust Rule -} For any coalition $S \subseteq \mathcal{N}$, the adjust rule is defined as follows:
\begin{align}
C^{S} = \textrm{Adjust}(S) = \begin{cases} S, & \mbox{if }  u(S)  =0,\\ \{S^{\textrm{MWC}},i_{1},\ldots,i_{k}\}, &
\mbox{if } u(S)  =1, \end{cases}
\end{align}
where $S^{\textrm{MWC}}$ is the MWC generated out of a winning coalition $S$ by simply excluding all players $i \in S$ such that $u(S)=u(S\setminus\{i\})=1$ and $\{i_{1},\ldots,i_{k}\}$ is the \emph{collection} of individual SUs $i_{j}$ that were excluded from $S$ by the adjust operation. %Hence, the output of an adjust operation is \emph{collection} $C^{S}$ such that if $S$ is \emph{losing} $C^{S}=\{S\}$ and if $S$ is \emph{winning} $C^{S}=\{S^{\textrm{MWC}},i_{1},\ldots,i_{k}\}$.
\end{definition}
The exclusion of the players from a winning coalition in order to yield an MWC can be done in any order. For the proposed CSS game, we adopt the following convention without any loss of generality:
\begin{convention}
The adjust operation generates an MWC from a winning coalition $S$ by excluding the players $i \in S$ such that $u(S)=u(S\setminus\{i\})=1$ (note that $u(\emptyset)=0$ by definition of any characteristic function), in a \emph{sequential} manner by an increasing order of their non-cooperative probability of miss $P_{m,i}$. In other words, the SUs verifying $u(S)=u(S\setminus\{i\})=1$ and having the smallest non-cooperative probability of miss are excluded first and their collection constitutes $\{i_{1},\ldots,i_{k}\}$.
\end{convention}
The motivation behind this convention is that, once and if excluded, SUs with a small non-cooperative probability of miss have higher chances (than SUs with larger miss probabilities) to find other partners with whom to form an MWC.

Having defined the necessary new concepts, the coalition formation Phase~2 of the proposed CF algorithm in Table~\ref{tab:alg1} can be easily modified by introducing the adjust operation. In this regard, Phase~2 of the algorithm in Table~\ref{tab:alg1} is split into two phases: Phase~2a and Phase~2b. In Phase~2a, an adjust operation is applied to every coalition in the network partition prior to merge-and-split, which results in a partition $\mathcal{T}_0$ that can be divided into two collections: a collection $\mathcal{W}_0$ with MWCs only and a collection $\mathcal{L}_0$ with the remaining coalitions. In Phase~2b, adaptive coalition formation through merge-and-split is applied only to partition $\mathcal{L}_0$. Nonetheless, in Phase~2b of the algorithm, the adjust operation needs to be further incorporated within merge-and-split. For instance, whenever a coalition forms by merge (or by split), it is subject to an adjust operation. In the event where a formed coalition is winning and the adjust operation results in an MWC, this MWC no longer participates in any future merge or split. However, any SUs excluded from the coalition through the adjust operation, continue to participate in merge-and-split. Consequently, the adaptive coalition formation process terminates by resulting in two collections: A collection $\mathcal{W}_f$ having only MWCs and another collection $\mathcal{L}_f$ with losing coalitions that can no longer be subject to any merge, split or adjust operations. The final network partition is simply $\mathcal{F} = \mathcal{W}_0 \cup \mathcal{W}_f \cup\mathcal{L}_f$. The convergence of the CF-PD algorithm to this final network is guaranteed since the presence of the adjust operation (which is an operation internal to each coalition) in Phase~2b does not affect the convergence of the merge-and-split iterations that was assessed in Section~\ref{sec:coalform} since this convergence is independent of the starting partition. This algorithm is referred to as the distributed coalition formation algorithm with probability of detection guarantees (CF-PD) algorithm and is summarized in Table~\ref{tab:alg2}. Note that, in the final network, the SUs that belong to the collection of losing coalitions $\mathcal{L}_f$ have improved their sensing performance compared to their non-cooperative state, however, they would need to wait environmental changes, such as mobility or the arrival of new SUs, in order to become winning SUs/coalitions, i.e., meet the required probability of detection constraint $\chi$ through cooperation.

\section{Distributed Implementation of CF and CF-PD}\label{sec:compl}
In this section, we discuss and present some of the challenges for the implementation of the proposed CF and CF-PD algorithms in practical cognitive networks.

\subsection{Neighbor Discovery}\label{sec:nd}
In the first phase of the CF and CF-PD algorithms, for discovering their neighbors, each coalition (or single SU) needs to utilize a neighbor discovery techniques in order to discover potential cooperation partners for performing the local merge. In other words, each coalition in the network surveys neighboring coalitions within merge range and attempts to merge according to the Pareto order. For performing merge, the SUs are required to know their own probabilities of miss and their distance to the neighbors in order to assess the potential utility that they can acquire by collaboration which is given by (\ref{eq:util}). Note that, although, in order to assess (\ref{eq:util}), the SUs need to also learn the probability of miss of their neighbors, i.e., the distance between the neighbors and the SU, this information can be sent by the neighbors over a control channel during the negotiation phase of the merge process. Hence, the SUs do not necessarily need to estimate this value by themselves.

In Rayleigh fading, the non-cooperative probability of miss is mainly a function of the average SNR to the PU, and, hence, the distance to the PU as seen in (\ref{eq:missprob}). Consequently, for performing coalition formation, the main information that needs to be gathered by the SUs consists of the distances between these neighbors as well as the distance between the PU and the SU. For finding this information, numerous methods are available for the SUs in practical cognitive networks:
\begin{enumerate}
\item Information received through control channels in cognitive networks such as the recently proposed Cognitive Pilot Channel~(CPC) \cite{CH01,CH03,CH04,CH05,CH06}.
   \item Signal-processing based methods \cite{CR02,NEW10,CS06,ZD00,PO00}.
\item Existing ad hoc routing neighbor discovery techniques applied in cognitive network \cite{ND00,ND01,ZD00,ADHOC00}.
\item Geo-location techniques \cite{CR01,SD00,DB00,DB01,DB02}.
\end{enumerate}
 \begin{table}[!t]
%\scriptsize
  \centering
  \caption{%\mycaption{%\vspace*{-1em}
    \vspace*{-0.2em}One round of the proposed CF-PD algorithm.}\vspace*{-1em}
    \begin{tabular}{p{8cm}}
      \hline
       \textbf{Initial State} \vspace*{.5em} \\
      \hspace*{1em} The network is partitioned by $\mathcal{T}=\{T_1,\ldots,T_k\}$ (At the beginning of all time $\mathcal{T}$ = $\mathcal{N}$ = $\{1,\ldots,N\}$ with non-cooperative SUs). \vspace*{.1em}\\
\textbf{Three phases in each round of the CF-PD algorithm} \vspace*{.5em}\\
\hspace*{1em}\emph{Phase 1 - Discovery:}   \vspace*{.1em}\\
\hspace*{2em}This phase remains the same as in the CF algorithm of Table~\ref{tab:alg1}.\vspace*{.1em}\\
\hspace*{1em}\emph{Phase 2a - Partition Adjustment:}   \vspace*{.1em}\\
\hspace*{3em}\textbf{Adjust the initial partition $\mathcal{T}=\{T_1,\ldots,T_k\}$}\vspace*{.2em}\\
\hspace*{4em}a) $C^{T_i}$ = Adjust($T_i$), $\ \forall \ T_i \in
\mathcal{T}$\vspace*{.2em}\\
\hspace*{4em}b) $\mathcal{T}_0 = \cup_{i=1}^{k}C^{T_i} $ can be divided into 2 collections\vspace*{.2em}\\
\hspace*{5em}b1) $\mathcal{L}_0=\{L_1,\ldots,L_p\},\ L_i \in \mathcal{T}_0,\ u(L_i)=0$ \vspace*{.2em}\\
\hspace*{5em}b2) $\mathcal{W}_0=\{W_1,\ldots,W_q\},\ W_i \in \mathcal{T}_0,\ u(W_i)=1$ \vspace*{.2em}\\
\hspace*{1em}\emph{Phase 2b - Adaptive Coalition Formation with Probability of} \vspace*{.1em}\\ \hspace*{6em}\emph{Detection Guarantees:}   \vspace*{.1em}\\
%\hspace*{2em}Coalition formation using merge-and-split with adjust occurs. \vspace*{.2em}\\
\hspace*{3em}\textbf{initialize} $\mathcal{L}_f=\mathcal{L}_0$ \vspace*{.2em}\\
\hspace*{3em}\textbf{repeat}\vspace*{.2em}\\
\hspace*{4em}a) $\{\mathcal{W}_f,\mathcal{L}_f\}$ = MergeAndAdjust($\mathcal{L}_f$); merge algorithm\vspace*{.2em}\\ \hspace*{5em}with adjust as explained in Section~\ref{sec:dcfpd}.\vspace*{.2em}\\
\hspace*{4em}b) $\{\mathcal{W}_f,\mathcal{L}_f\}$ = SplitAndAdjust($\mathcal{L}_f$); split algorithm with \vspace*{.2em}\\
 \hspace*{5em} adjust as explained in Section~\ref{sec:dcfpd}.  \vspace*{.2em}\\
\hspace*{3em}\textbf{until} merge-and-split with adjust terminates.\vspace*{.2em}\\
\hspace*{2em}Final partition is  $\mathcal{F} = \mathcal{W}_0 \cup \mathcal{W}_f \cup\mathcal{L}_f$ \vspace*{.2em}\\
  \hspace*{1em}\emph{Phase 3 - Coalition Sensing:}   \vspace*{.1em}\\
\hspace*{2em}This phase remains the same as in the CF algorithm of Table~\ref{tab:alg1}.\vspace*{.1em}\\
\textbf{Similar to the CF algorithm, Phase~2 (a and b) of CF-PD is repeated periodically to allow the SUs to adapt the network structure to environmental changes such as mobility. } \vspace*{.4em}\\
   \hline
    \end{tabular}\label{tab:alg2}\vspace{-0.6cm}
\end{table}
\subsubsection{Neighbor discovery through control channels}
 In order to maintain a conflict-free environment between the SUs and PUs of a cognitive network, there has been a recent interest in providing, through control channels, assistance to the SUs on both the spectrum sensing and access levels by providing useful information that the SUs can use to improve their performance. For example, in \cite{CH01}, a control channel, called the Common Spectrum Coordination Channel~(CSCC) has been proposed for the announcement of radio and service parameters to allow a coordinated coexistence between SUs and PUs. Following a similar concept, the Cognitive Pilot Channel~(CPC) has been recently proposed \cite{CH03,CH04,CH05,CH06} as a means for providing frequency and geographical information to cognitive users, assisting them in sensing and accessing the spectrum. Notably, the CPC can assist the SUs to discover neighboring PUs and SUs in the presence of multiple access technologies and heterogeneous environments. Thus, as explained in \cite{CH03,CH04,CH05}, the CPC is a control channel that carries information such as available spectrum opportunities, existing PUs or geographical information on neighboring SUs. A CPC channel can be deployed in a cognitive network by either a broadcast manner, i.e., a station that broadcasts the CPC to all of its users, or an on-demand manner, i.e., the CPC information is sent  to the SUs upon request. In this paper,  the CPC can be delivered to the SUs by the receivers or a stand alone spectrum broker which is able to obtain information on the existing SUs, PUs, or other needed CPC details. Following the formation of a coalition, an SU can also signal to its receiver or to the broker its membership in a coalition, which is an information that can be further appended to the CPC.

Consequently, by listening to the CPC transmission the SUs can discover their own probabilities of miss (i.e., their distances to the PU), the presence of neighboring SUs, the distance to other SUs, and, possibly the coalition memberships of these SUs. Since various efficient implementations of the CPC or other information control channels for cognitive networks are ubiquitous, then, using these channels, the SUs can acquire the needed information for forming coalitions through CF or CF-PD with a reasonable overhead.

 \subsubsection{Signal processing-based methods} \label{sub:sig}
 In Rayleigh fading, discovering neighboring nodes and finding an estimate of the distances to these nodes can be achieved through signal processing techniques \cite{PO00} by computing the average received power, from the neighboring SUs and from the PU. To do so, in the proposed model, prior to any coalition formation, i.e., in Phase~1 of the algorithm, each SU can attempt to compute the average received power, while listening to a control channel for example. Another approach would be for the SUs to estimate this average received power during the operation of the network prior to coalition formation. In other words, in a given initial partition $\mathcal{T}_0$, the SUs start by performing spectrum sensing and, then, use the results in order to access the channel. During spectrum access, the SUs can monitor the transmission received from other SUs, and, then, discover existing SUs in the network as well as the distance to these SUs through an estimation, using signal processing-based methods, of the average received SNR from these SUs. This information on the neighbors and the distances to them will be useful for the SUs to take adequate coalition formation decisions in Phase~2 of the algorithm. It must be stressed that the use of received power estimates for neighbor discovery is quite common as shown in \cite{CR02,NEW10,CS06,ZD00,PO00}.

The measured mean received power from the PU (whose transmission is always monitored by the SUs for cognitive radio access), can be used by every SU to find its individual probability of miss. For instance, in this paper, we deal mainly with the probability of miss averaged over the small-scale Rayleigh fading as given in (1), and which is common in many cognitive radio literature, e.g.  \cite{CS00,NEW04,CS01,CS04,CS03,CS02,CS05,CS06} and the references therein. By closely inspecting (1), one would see that, for a given time-bandwidth product $m$ and detection threshold $\lambda$, the probability of miss of an SU mainly depends on the average received SNR from the PU signal, which, is a function of the path loss, and, thus, the distance between the SU and PU. Typically, an SU needs only to average the received power from the PU over several samples to compute this average (with respect to fading) SNR. Therefore, an SU can use the measured mean received power from the PU to estimate its distance to the PU (or the average SNR), and, subsequently, its individual probability of miss. For a coalition of SUs, each member can find its own individual probability of miss and, consequently, since the SUs of a coalition are aware of the location of each other, they can work out the potential probability of miss of their coalition as per (3).

Finally, note that, using signal-processing techniques to discover neighboring SUs, a coalition of SUs $T_i \in \mathcal{T}$ needs only to discover SUs in a small geographical area around it (since otherwise no potential merge is possible as demonstrated for example in Theorem~1 for the 2-SU case) and does not need to know which coalitions each SU belongs to. The information on the membership of the SUs into a coalition can be found out in subsequent pairwise interactions over a control channel between the coalition $T_i$ and its discovered neighboring SUs (since each SU knows which coalition it belongs to).

\subsubsection{Existing ad hoc routing discovery techniques applied in cognitive networks}
In the past decade, techniques for neighbor discovery have been widely spread in the context of ad hoc routing \cite{ADHOC00} whereby numerous methods for finding neighboring nodes and information on these nodes are developed. Recently, these approaches have been leveraged in order to provide neighbor discovery in cognitive radio networks \cite{ND00,ND01,ZD00}.

While an in-depth study of these neighbor discovery techniques is beyond the scope of this paper, we will provide a simple example of such techniques which can assist the SUs for the purpose of coalition formation. For instance, one popular idea for neighbor discovery is the use of methods based on the well known on demand distance vector (AODV) protocol of ad hoc routing \cite{ZD00}. In this protocol, each node periodically broadcasts a ``Hello'' packet, in order to detect its neighbors. In a cognitive network, several methods based on this idea can be used by the SUs to discover the PUs or other SUs. As shown in \cite{ZD00}, to do so, each SU can establish an ad hoc temporary control information channel, over which it broadcasts its ``Hello'' packet which contains the SU's mobile IP address, location information, frequencies of all ad hoc temporary control channels, exact time when SU will tune to ad hoc channels to listen for reception, transmission schedules of ``Hello'' messages,  time stamp according to its internal clock, coalition membership, or other information. Several techniques for transmitting the ``Hello'' packets are given in \cite{ZD00}. Upon receiving the ``Hello'' message, an SU can adjust its own transmission power according to the distance calculated from the information contained
within the broadcasted ``Hello'' message, and then broadcast a ``Hello Ack'' message with its own information. Using this ``Hello'' and ``Hello ACK'' procedure, the SUs can discover each other \cite{ZD00} and, subsequently, engage in a coalition formation process.

Beyond this idea, numerous other techniques for cognitive radio neighbor discovery, using similar concepts from ad hoc routing discovery, can be found in \cite{ND00,ND01} (and references therein) along with their details. By adopting any of these techniques, the neighbor discovery phase of the proposed coalition formation algorithm can be implemented, allowing the SUs to gather any needed information for the merge-and-split phase.

\subsubsection{Geo-location techniques}
 In some cognitive networks, in order to improve the coexistence between the secondary and primary networks, the SUs are required to register and connect, through the Internet, to incumbent databases \cite{SD00,DB00,DB01,DB02}. The principal purpose of this requirement is to provide a mechanism to inform SUs about their neighboring SUs or PUs \cite{SD00}. These databases can hold information on the PUs, the SUs, their positions, their capabilities, and so on. For the purpose of coalition formation, the databases may also hold a field that displays the coalition membership of the SUs, if any.

The information in these databases is entered either automatically by every SU or PU that enters the network, or by a centralized administrator that controls these databases \cite{SD00}. In this context, the need for knowing the position of the nodes (SU or PU) prior to filling the database is required. For fixed SUs or PUs, obtaining this information for filling the databases is straightforward. However, when the SUs or PUs are mobile devices, two cases can be distinguished. In an outdoor environment, the SUs or PUs can use techniques based on the global positioning service (GPS) in order to obtain their own locations and update the databases. If no GPS is available or the nodes are indoor, finding the position becomes more challenging. In this case, the nodes can use advanced positioning techniques such as reference anchors \cite{SD00}, the signal processing techniques of Subsection~\ref{sub:sig}, or other techniques \cite{DB00,DB01,DB02}.

Although the technical details on building the databases is beyond the scope of this paper (the interested reader is referred to \cite{SD00,DB00,DB01,DB02} and the references therein), nonetheless, for coalition formation, the SUs can access the databases in order to discover the position of the PU or their neighboring SUs, or even, the coalition memberships of the neighbors. By doing so, the SUs would obtain all the necessary information for the coalition formation phase of CF or CF-PD.

In summary, neighbor discovery techniques for cognitive radio networks are abundant. For the proposed coalition formation algorithm, any of these techniques can be adopted in Phase~1 as they all allow the SUs to obtain estimates of the needed information for the coalition formation phase. Certainly, there is a tradeoff between accuracy and complexity for each of these techniques, however, for coalition formation, one can adopt any technique with a reasonable tradeoff (e.g., the CPC or the signal processing techniques) as long as it can convey enough information on the neighboring SUs and the PU in the network.

% can utilize several algorithms for neighbor discovery that have already been developed and tailored for ad hoc cognitive network such as in \cite{ND00} or \cite{ND01}. These algorithms are inherently based on techniques similar to those used in the ad hoc routing discovery. Other approaches such as the cooperative neighbor discovery techniques used in \cite{CS06,NEW10} can also be used at this stage.
%
%In fact, cognitive radio devices are required to be connected through the Internet to incumbent databases that can provide location information on their neighbors as well as on the PU . In our case, these databases can be used by the SUs for finding the  distances required for assessing the possibility of coalition formation using our proposed algorithms.

   %Once each SU estimates its own non-cooperative probability of miss
%\footnote{During a small time period prior to coalition formation, each SU can assess its non-cooperative probability of miss based on its individual observations of the PU.}
Hence, once these parameters are estimated using either the received power estimation, or geo-location techniques, or the CPC, the coalitions (or individual SUs) can then decide on the partners, if any, with whom they can cooperate (merge) for spectrum sensing. Moreover, with regards to the split operation, each formed coalition can internally make the decision to split or not depending on whether its members can find a split form that is preferred by the Pareto order.
\subsection{Coalition Formation Operations: Merge, Split, and Adjust}
 The complexity of the coalition formation phase of the proposed CF algorithm lies mainly in the complexity of the merge and split operations. It is important to stress that the complexity of the algorithm is independent of whether a $\mathbb{D}_c$-stable partition exists or not. For instance, as per Lemma~1 and as shown in \cite{KA01}, any implementation of a merge-and-split algorithm would reach the $\mathbb{D}_c$-stable partition  (if it exists) at no extra complexity).  Hence, reaching an optimal and stable partition (when such a partition exists) does not incur extra complexity. For a given network structure, during one merge process, each coalition attempts to merge with other coalitions in its vicinity in a pairwise manner. In the worst case scenario, every SU, before finding a suitable merge partner, needs to make a merge attempt with all the other SUs in $\mathcal{N}$. In this case, the first SU requires $N-1$ attempts for merge ($N$ being the number of SUs), the second requires $N-2$ attempts and so on. The total worst case number of merge attempts will be $\sum_{i=1}^{N-1}i= \frac{N(N-1)}{2}$. In practice, the merge process requires a significantly \emph{lower} number of attempts since finding a suitable partner does not always require to go through all the merge attempts (once a suitable partner is identified the merge will occur immediately). This complexity is further \emph{reduced} by the fact that a coalition does not need to attempt to merge with far away or large coalitions where the false alarm cost is large relatively to the benefit or where the false alarm is larger than $\alpha$ (Theorem~1 being an example of such merge regions for the two SUs case). Moreover, after the first run of the algorithm, the initial $N$ non-cooperative SUs will self-organize into larger coalitions. Subsequent runs of the algorithms will deal with a network composed of a number of coalitions that is much smaller than $N$; reducing the merge attempts per coalition. Finally, once a group of coalitions merges into a larger coalition, the number of merging possibilities for the remaining SUs will decrease as the merge region becomes smaller (due to cooperation costs). % In summary, although in worst case scenarios
%the merge process requires around $\frac{N\cdot(N-1)}{2}$ attempts, in practice, the process is far less complex.

For the split rule, in a worst case scenario, splitting  can involve finding all the
possible partitions of the set formed by the SUs in a single coalition. For a given coalition, this number is given by the Bell number which grows exponentially with the number of SUs in the coalition (here the partition search is restricted to each coalition and is not over the whole set $\mathcal{N}$) \cite{BN00}. In practice, this split operation is restricted to the formed coalitions in the network, and, thus, it will be applied on \emph{small} sets. Due to the cost function in (\ref{eq:util}) and (\ref{eq:logbarr}) , the size of each coalition resulting from CSS coalition formation is limited by the increasing false alarm cost as well as by the SUs' locations, thus, the coalitions formed in the proposed CSS game are small. As a result, the split complexity is limited to finding possible partitions for small sets which can be reasonable in terms of complexity. The split complexity is further reduced by the fact that, in most scenarios, a coalition does not need to search for all possible split forms. For instance, once a coalition identifies a split structure verifying the Pareto order, the SUs in this coalition will split, and the search for further split forms is not needed. In summary, the practical aspects dictated by the cognitive network such as the increasing false alarm cost, the errors on the reporting channel, the SUs' locations and the sequential split search will significantly diminish the split complexity. In addition, the complexity of both merge and split is also reduced due to the fact that the algorithm is distributed, and hence all decisions are local to the coalitions.

In CF-PD, the adjust operation can be decided by each coalition internally similar to the merge or split operations as explained in Section~\ref{sec:compl}. Further, for the CF-PD algorithm, typically the merge and split operation will be less complex than for the CF algorithm in Table~\ref{tab:alg1} since winning coalitions do not participate in the merge and split process. However, for CF-PD, another operation comes into place, which is the adjust operation. The adjust operation is \emph{local} to each \emph{winning} coalition $S$ in the network (adjust has no effect on losing coalitions) and its complexity grows linearly in $|S|$ (number of SUs in $S$). Generally, the coalitions that form in the network are small, thus, the adjust rule complexity is very small.

\subsection{Coalition Based Sensing}
Once the coalitions form, the operation of the network within every coalition (Phase~3 in the CF and CF-PD algorithms) is, as briefly prescribed in Sections~\ref{sec:prob} and \ref{sec:mergeandsplit}, systematic. First and foremost, within every coalition each SU prepares a \emph{single} sensing bit based on its local observation, and then delivers this bit to the coalition head through the wireless channel. Subsequently, the coalition head can perform the fusion of the bits and broadcast back the result, in a single bit, to all the coalition members. In this paper, we adopt the common OR-rule for decision fusion (hard decisions) that is widely used in collaborative sensing in cognitive radio \cite{CS00,CS01,CS04,CS03,CS02,CS05,CS06}\footnote{Note that, the proposed coalition formation algorithm can be tailored to accommodate other decision fusion techniques.}. The operation of the network in this phase, i.e., Phase~3, can be used in order to compare the communication overhead for collaborative sensing between our proposed approach with that of the traditional network-wide fusion center solution such as \cite{CS00,CS01,CS04,CS03,CS02,CS05,CS06}. The reason is that, the actual collaborative sensing which takes place in Phase~3 while the coalition formation phase occurs only periodically (seldom in moderately changing environments). For example, in a static network, coalition formation is performed only once, and subsequently, the SUs can perform collaborative sensing at the level of their coalitions. In this regards, in the network-wide fusion center, every single SU in the network needs  to report its sensing bit to the fusion center, which yields up to $N$ bits report where $N$ is the total number of SUs. When the network is large and the SUs are spread out far way from the fusion center, this communication overhead is quite significant. In contrast, in Phase~3 of the proposed algorithm, in order to perform coalition-based collaborative sensing, the SUs belonging to the same coalition, would need to only do a \emph{local exchange} of the sensing bits, at the level of a coalition, which would have a size much smaller than $N$. For instance, as demonstrated in Theorem~2, due to the false alarm constraint, the size of the coalitions resulting from the proposed algorithm is upper bounded. In fact, as will be seen through simulations in Section~\ref{sec:sim}, the size of the resulting coalitions will, in general, be much smaller than $N$.  Thus, clearly, the bandwidth required for the communications of the sensing bits in Phase~3 of our algorithm is quite low, as only single bits are exchanged between coalition members. In fact, within each coalition, the overhead communication required for performing collaborative sensing will be much smaller than $N$. %Furthermore, due to the fact that the network resulting from our coalition formation algorithm consists of nearby coalitions, the communication overhead for exchanging these sensing bits within every coalition is significantly lower than existing solutions which employ a network-wide centralized fusion center that gathers sensing bits from all SUs in the network.
Beyond the overhead communication advantage, performing collaborative sensing at the level of coalitions would save power and energy for the SUs since, in general, they would report their bits to a neighboring coalition head. Moreover, the proposed approach would make reporting the sensing bits more reliable since, in each individual formed coalition, the probability of error over the reporting channel as per (\ref{eq:err}) would be much smaller than in the case where all $N$ SUs need to submit their bits to a fusion center. In a nutshell, performing collaborative sensing at the level of coalitions not only reduces the overhead communication but, also, improves reliability and reduces the power consumption (used for sensing) of the SUs.  For multiple access within each coalition (during sensing bit exchange), well known existing MAC protocols can be readily used, such as simple orthogonal channels \cite{CS00,CS01,CS04,CS03,CS02,CS05,CS06}, random access techniques \cite{SD00}, or low temperature handshake solutions such as in \cite{ZH00}. For adapting to environmental changes such as mobility, the SUs perform the first two phases of the proposed CF algorithm in Table~\ref{tab:alg1} \emph{periodically} over time (see Subsection~\ref{sec:mergeandsplit}).

\subsection{Impact of Potential Estimation Errors}
In this paper, as mentioned in Section~\ref{sec:prob}, all the SUs are considered as honest nodes that do not exhibit any malicious or cheating behavior. Hence, only estimation errors (e.g., due to a low received SNR from the PU) may have some impact on the formed coalitions, if, once a coalition forms, the probability of miss of an SU is not as good as envisioned during the coalition formation process. Here, all the cooperation decisions are based on the probability of miss that is averaged over the fading realizations, as per (\ref{eq:missprob}). In this context, the probability of miss in (\ref{eq:missprob}) already captures the possibility of having a low SNR received from the PU, which, in turn, maps into an estimation error. Therefore, the SUs that are more apt to making estimation errors, i.e., SUs far away from the PU, would exhibit a high non-cooperative probability of miss as per (\ref{eq:missprob}). In consequence, whenever such SUs report their probabilities of miss to potential merge candidates, these candidates would be aware of the high probability of miss and would use it in order to find out whether the cooperative performance, i.e., as per (\ref{eq:missclu}), would satisfy the Pareto order or not. Hence, the estimation errors due to low SNR, as obtained by (\ref{eq:missprob}), are handled adequately through the Pareto order merge or split decisions that are based on the utility function in  (\ref{eq:util}). Therefore, these errors do not perturb the convergence of the merge-and-split algorithm.

Nonetheless, following the convergence of merge-and-split, consider a group of SUs that have merged together into a coalition $S$. This coalition, in Phase~3 of the algorithm, would start acting cooperatively and performing distributed collaborative sensing. During this phase, the SUs in $S$ can monitor their overall performance as opposed to the expected performance when the merge agreement has occurred. In this case, if it turns out that, when the SUs in $S$ operate cooperatively, the actual probability of miss achieved by one of the SUs in $S$ is, due to estimation errors, much higher than the value potentially envisioned during coalition formation (e.g., this can be estimated through the number of collisions that the SUs in this coalition had with the PU when they start accessing the spectrum), then this SU can be excluded by the coalition. The procedure for excluding these SUs is out of the scope of the present paper; however, one can use approaches similar to the ones used for identifying malicious SUs (e.g., onion peeling) such as in \cite{HL00,HL01}.

\section{Simulation Results and Analysis}\label{sec:sim}
For simulations, the following network is set up: The PU is
placed at the center of a $3$~km $\times 3$~km square area with the SUs randomly deployed in the area around the PU. We set the PU transmit power $P_{\textrm{PU}}=100$~mW, the SU transmit power for reporting the sensing bits $P_i=10$~mW, $\forall i \in \mathcal{N}$ and the noise $\sigma^2=-90$~dBm. For path loss, we set $\mu=3$ and $\kappa = 1$. The maximum false alarm constraint is set to $\alpha = 0.1$, as recommended by the IEEE 802.22 standard \cite{RE00}.  In practical cognitive networks, it is desirable to reduce the access delays and the bandwidth needed for sensing \cite{CR02,NEW04}, hence, the  the time bandwidth product $m$ is generally chosen to be small \cite{CR02,NEW04,CS00,CS01,CS03,CS04}. In the simulations, unless stated otherwise, we select $m=5$ which is a common value used widely in all literature \cite{CS00,NEW04,CS01,CS03,CS04}. In this section, all statistical results (when applicable) are averaged over $5000$ random locations of the SUs as well as a range of energy detection thresholds $\lambda$ that do not violate the false alarm constraint; this in turn, maps into an average over the non-cooperative false alarm range $P_f \le \alpha$ (obviously, for $P_f > \alpha$ no cooperation is possible).

\begin{figure}[!t]
\begin{center}
\includegraphics[angle=0,width=90mm]{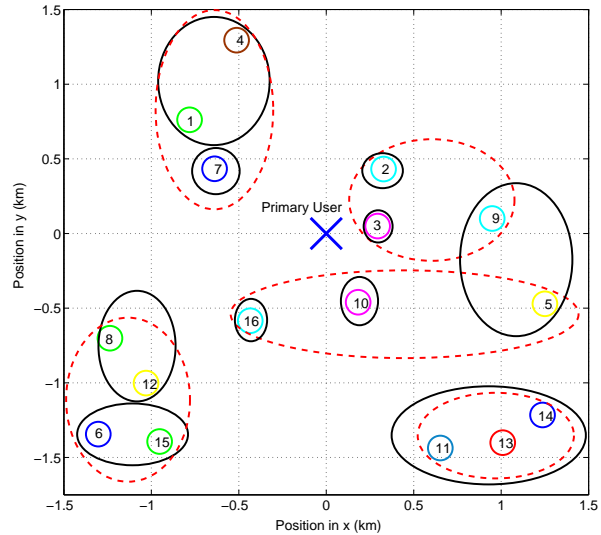}
\end{center}\vspace{-0.5cm}
\caption {A snapshot of the coalitional structure resulting from both CF (in dashed line) and CF-PD (solid line) algorithms for $N=16$ SUs, $P_f=0.01$, and a target detection probability  $\chi=95\%$ for CF-PD.} \label{fig:snapshot}\vspace{-0.4cm}
\end{figure}

In Fig.~\ref{fig:snapshot}, we show a snapshot of the network structure resulting from the proposed CF  (dashed line) and CF-PD (solid line) algorithms for $N=16$ randomly placed SUs, a non-cooperative false alarm $P_f = 0.01$, and  a target detection probability $\chi=95\%$ for CF-PD. We note that, while the CF-PD yields $10$ coalitions which comprise $5$~SUs that decided to remain non-cooperative, the CF algorithm yields only $5$ coalitions with no SUs remaining non-cooperative. Due to their proximity to the PU, the non-cooperative detection probabilities of SUs $2,3,7,10,$ and $16$ are, respectively, $98.8\%,99.8\%,96.6\%,99.1\%$, and $97.2\%$, which satisfy the required target threshold $\chi=95\%$ for these SUs. Thus, in CF-PD, SUs $2,3,7,10$ and $16$ form a singleton MWC each and have no incentive to cooperate as they already satisfy the desired performance non-cooperatively. Moreover, coalitions $\{1,4\}$, $\{8,12\}$, and $\{6,15\}$ are MWCs, and, thus, do not cooperate any further during CF-PD. In contrast, when the CF algorithm is applied, the SUs seek to minimize their probability of miss and, thus, SUs $2$ and $3$ form a coalition with SU $9$, SUs $10$ and $16$ form a coalition with SU $5$, while SU $7$ minimizes its probability of miss by merging with coalition $\{1,4\}$. Further, for CF, SUs $6,8,12$, and $15$ minimize their probability of miss by forming a single coalition $\{6,8,12,15\}$. Finally, we note that for both CF and CF-PD coalition $\{11,13,14\}$ forms, and this coalition is a losing coalition since the achieved probability of detection by its member SUs is only $89.5\%$ (due to the location of its members). Moreover, this coalition is unable to find any suitable partners that can help it improve its utility further in CF, or meet the target detection threshold in CF-PD. In a nutshell, Fig.~\ref{fig:snapshot} illustrates how, in a practical cognitive network, different coalitions can emerge depending on the incentives and cooperative objectives of the SUs.

\begin{figure}[!t]
\begin{center}
\includegraphics[angle=0,width=80mm]{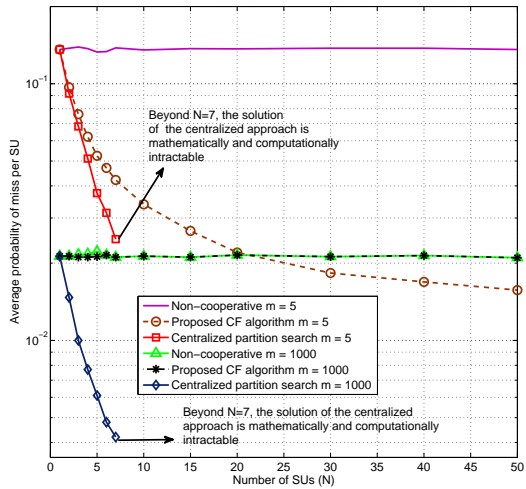}
\end{center}\vspace{-0.7cm}
\caption {Average probability of miss versus number of SUs $N$ for $m=5$ and $m=1000$.} \label{fig:perf}\vspace{-0.5cm}
\end{figure}
Figs.~\ref{fig:perf} and \ref{fig:falarm} show, respectively, the average probabilities of miss and the average false alarm probabilities achieved per SU for different network sizes and for a typical time-bandwidth product of $m=5$ and a large time-bandwidth product of $m=1000$. In Fig.~\ref{fig:perf}, we show that, for $m=5$, the proposed CF algorithm yields a significant advantage in the average probability of miss reaching up to $88.45\%$ reduction (at $N=50$) relative to the non-cooperative case which is an order of magnitude of improvement. This advantage is increasing with the network size $N$. In this figure, we also show the results of a centralized exhaustive search solution which minimizes the average probability of miss per SU, subject to the false alarm constraint $\alpha$. This solution is shown for up to $N=7$ since it is mathematically intractable for larger networks.

We note that, for all $m$, a gap exists between the performance of the CF algorithm and that of the centralized solution. This gap stems mainly from the individual choices of the SUs when they act in a distributed manner, i.e., selecting their partners based on a balance between gains and costs as opposed to a centralized approach that does not capture these individual incentives. For CF, these incentives are captured by the cost function in (\ref{eq:logbarr}) which increases drastically when the false alarm probability is in the vicinity of $\alpha$. This increased cost makes it generally non-beneficial for coalitions with high false alarm levels, e.g., false alarm values close to $\alpha$, to collaborate in the distributed CF approach as they require a large probability of miss improvement to compensate the cost in their utility (\ref{eq:util}). As a result, even though the CF algorithm yields a performance gap in terms of miss probability for all $m$, the individual decisions of the SUs force a false alarm for the distributed case smaller than that of the centralized solution as seen in Fig.~\ref{fig:falarm} at all $m$.  In other words, when an SU (or coalition) is taking its own decision, it needs to obtain a gain, in terms of probability of miss, that is relatively larger than the increase in the false alarm (relative to the constraint $\alpha$) as captured by the utility and cost function in (\ref{eq:util}). In contrast, in the centralized approach, the optimization problem stated in Section~\ref{sec:prob} does not involve the individual utilities or decisions of the SUs, instead it optimizes the overall average probability of miss given a constraint on the overall false alarm probability. Consequently, the centralized approach would definitely yield a smaller average probability of miss, but it also yields a false alarm level that is much closer to the constraint $\alpha$, i.e., a false alarm level that is higher than in the distributed approach. This result is corroborated in Fig.~\ref{fig:falarm} where, at all $m$, the achieved average false alarm by the proposed distributed solution outperforms that of the centralized solution, however, it is still outperformed by the non-cooperative case. %Thus, while the centralized solution achieves a better probability of miss; the CF algorithm compensates this performance gap through the average achieved false alarm (due to SUs' incentives)

In addition, even though a small $m$ is typical in cognitive networks \cite{CS00,NEW04,CS01,CS03,CS04}, a large time-bandwidth product may be used in some instances when detection at low SNR is required. By increasing the observation time (or the bandwidth used for sensing), the SUs can improve their non-cooperative average probability of miss non-cooperatively as shown in Fig.~\ref{fig:perf} for $m=1000$. However, this gain comes at the expense of a long delay for spectrum access or a wasted bandwidth. At $m=1000$, although for the centralized approach a performance improvement in terms of probability of miss is possible, we note that for the proposed distributed CF algorithm no coalitions are formed and the performance is comparable to the non-cooperative case. This is mainly due to the fact that, at $m=1000$, the average non-cooperative probability of miss of all SUs is already quite small (around $2\%$), and, hence, when the SUs act in a distributed manner, they find little improvement in their utility (\ref{eq:util}) through cooperation, since the gains in probability of miss will be accompanied with a large increase in the false alarm as shown in Fig.~\ref{fig:falarm} for the centralized approach at $m=1000$. However, by comparing the results at $m=5$ and $m=1000$, Fig.~\ref{fig:perf} shows that, by using our proposed CF algorithm, the SUs can improve their detection significantly without the need for a wasted bandwidth or long access delays. For example, for a network of $N=50$ SUs at $m=5$, the SUs can achieve a probability of detection which is comparable (and slightly outperforms) the non-cooperative case at $m=1000$. Hence, the comparison between the CF results at $m=5$ and the non-cooperative results at $m=1000$ in Fig.~\ref{fig:perf} demonstrates that, instead of increasing the sensing observation time or bandwidth (notably when detection at low SNR is needed), the SUs can utilize the proposed CF algorithm to improve their detection performance while saving bandwidth and reducing their spectrum access delay.

%In other words, for false alarm probabilities per coalition close to the constraint $\alpha$ through %the proposed utility and cost models in (\ref{eq:util}) and
 %(\ref{eq:logbarr}), it becomes harder for SUs to cooperate using the distributed as the cost as implied by the barrier function is very large,

\begin{figure}[!t]
\begin{center}
\includegraphics[angle=0,width=85mm]{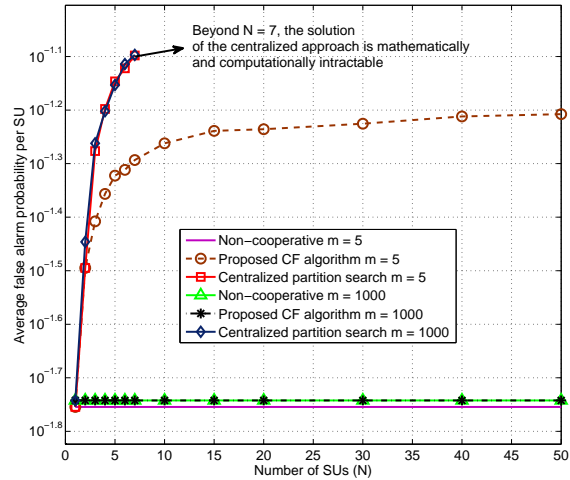}
\end{center}\vspace{-0.5cm}
\caption {Average false alarm probability versus number of SUs $N$ for $m=5$.} \label{fig:falarm}\vspace{-0.6cm}
\end{figure}

In Fig.~\ref{fig:perffal}, we show the average probabilities of miss per SU for different thresholds $\lambda$ (range of non-cooperative false alarm $P_f \le \alpha$) for $N=7$. Fig.~\ref{fig:perffal} shows that, as $P_f$ decreases, the performance advantage of CSS for both the centralized and distributed CF solutions increases (except for very small $P_f$ where the advantage in terms of miss probability reaches its maximum). The performance gap between centralized and distributed is once again compensated by a false alarm advantage for the distributed CF solution as already seen and explained in Fig.~\ref{fig:falarm} for $N=7$. Finally, Fig.~\ref{fig:perffal} shows that as $P_f$ goes to $\alpha = 0.1$, the CSS advantage diminishes as the network tends to the non-cooperative case.

%\begin{table}[!t]
 %   \begin{center}
  %          \caption{\label{tab:merg} \vspace{-0.1cm}
%Analysis of the formation of a coalition of size $2$ between a SU and three classes
%of SUs using Theorem~\ref{th:one}.}
 %       \begin{tabular} {|c|c|c|c|c|}
  %          \hline
   %       &\multicolumn{2}{c|}{$\tilde{P}_{e,j,i}$}&\multicolumn{2}{c|}{ $\tilde{d}_{i,j}$ (km)}\\
    %   \hline
%SU~1 merge with &Approx.&Exact&Approx.&Exact\\
 %\hline
%Class 1 SU, $P_{m,C_1}\!=\!0.2156$&0.0779&0.0781&1.5917&1.5933\\
  %\hline
%Class 2 SU, $P_{m,C_2}\!=\!0.5459$&0.0715&0.0718&1.5352&1.5375\\
 % \hline
 %Class 3 SU, $P_{m,C_3}\!=\! 0.7756$&0.0563&0.0595&1.3919&1.4234\\
 %\hline
  %   \end{tabular}  \vspace{-0.2cm}

   % \end{center}
%\end{table}
In Fig.~\ref{fig:mergef} we investigate, through Theorem~\ref{th:one} (with $P_f =0.01$), the merge possibilities between two secondary users $\textrm{SU}_1$ and $\textrm{SU}_2$, with $P_{m,1} \le P_{m,2}$. The results are shown for different locations of $\textrm{SU}_1$ when $\textrm{SU}_2$ belongs to circles centered at the PU and with different radii (each circle radius represents a class of non-cooperative probability of miss). Fig.~\ref{fig:mergef} shows the angle, in degrees, of the sector centered at $\textrm{SU}_1$ and which intercepts the merge region arc at the various circles where $\textrm{SU}_2$ belongs (obviously a larger angle intercepts a larger merge region arc). The angle is obtained by a simple geometrical computation after finding the minimum distance $\tilde{d}_{1,2}$ between $\textrm{SU}_1$ and $\textrm{SU}_2$ using Theorem~1. First and foremost, Fig.~\ref{fig:mergef} shows that the approximate solution obtained by Theorem~1 is almost the same as the exact numerical (graphical) solution, which corroborates that the approximation in Theorem~1 for the distances is tight. Moreover, Fig.~\ref{fig:mergef} clearly shows that, for the different $\textrm{SU}_1$ locations, the merge region is larger as the radius
\begin{figure}[!t]
\begin{center}
\includegraphics[angle=0,width=80mm]{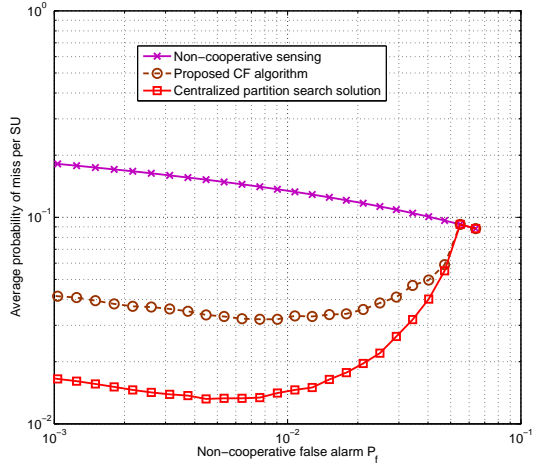}
\end{center}\vspace{-0.5cm}
\caption {Average probabilities of miss per SU versus non-cooperative false alarm $P_f$ (or energy detection threshold $\lambda$) for a network of $N=7$ SUs with $m=5$.} \label{fig:perffal}\vspace{-0.6cm}
\end{figure}
\begin{figure}[!t]
\begin{center}
\includegraphics[angle=0,width=80mm]{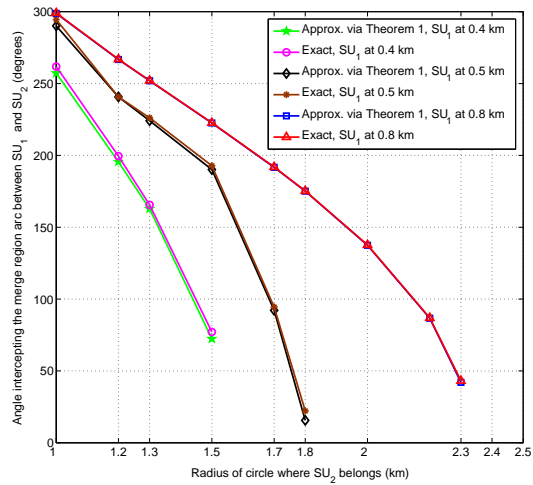}
\end{center}\vspace{-0.5cm}
\caption {Approximate (via Theorem~\ref{th:one}) and exact (graphical solution for inequality) angle centered at $\textrm{SU}_1$ and intercepting the merge region arc (for forming a coalition of size $2$) between $\textrm{SU}_1$ and $\textrm{SU}_2$ ($P_{m,1} \le P_{m,2}$) when $\textrm{SU}_2$ belongs to circles centered at the PU and having different radii (each circle radius represents a different non-cooperative probability of miss) and for different locations of $\textrm{SU}_1$.} \label{fig:mergef}\vspace{-0.6cm}
\end{figure}
of $\textrm{SU}_2$'s circle is smaller (i.e. error on the reporting channel is smaller). Further, we note that as $\textrm{SU}_1$ becomes closer to the PU, the number of circles where $\textrm{SU}_2$ can belong and on which a merge region exists becomes smaller. For instance, when $\textrm{SU}_1$ is at $0.8$~km, it can merge with $\textrm{SU}_2$ moving on circles of radii up to $2.3$~km. In other words, when $\textrm{SU}_2$ belongs to a circle with a radius $1.5$~km larger than the radius of the circle to which $\textrm{SU}_1$ belong, the coalition cannot form. However, when $\textrm{SU}_1$ is at $0.4$~km it can merge with $\textrm{SU}_2$ belonging to circles with radii up to $1.5$~km, hence, only around $1.1$~km larger than the radius of the circle to which $\textrm{SU}_1$ belong. This result shows that, when $\textrm{SU}_1$ has a smaller probability of miss (i.e. better non-cooperative performance), the collaboration possibilities with distant SUs become smaller, since $\textrm{SU}_1$ is already satisfied with its non-cooperative probability of miss and does not have an incentive to collaborate with very distant SUs. In summary, Fig.~\ref{fig:mergef} provides interesting insights on how two SUs can merge, depending on their location with respect to the PU.

Fig.~\ref{fig:mobconv} shows how the structure of the cognitive network with $N=50$ SUs evolves and self-adapts, while all the SUs use a basic random walk mobility model whereby the nodes move at a constant speed of of $120$~km/h in a random direction uniformly distributed between $0$ and $2\pi$, over periods of $5$~seconds for a total duration of $5$~minutes. The proposed CF algorithm is repeated periodically by the SUs every $\theta=5$ seconds, in order to provide self-adaptation to mobility. As the SUs move, the structure of the network changes, with new coalitions forming and others splitting. The network starts with a non-cooperative structure made up of $50$ independent SUs. In the first step, the network self-organizes
in $18$ coalitions with an average of $2.77$ SUs per coalition. As time evolves, the SUs self-organize the structure which is changing as new coalitions form or others split. After the $5$ minutes have elapsed the final structure is made up of $20$ coalitions with an average of $2.5$ SUs per coalition.

\begin{figure}[!t]
\begin{center}
\includegraphics[angle=0,width=80mm]{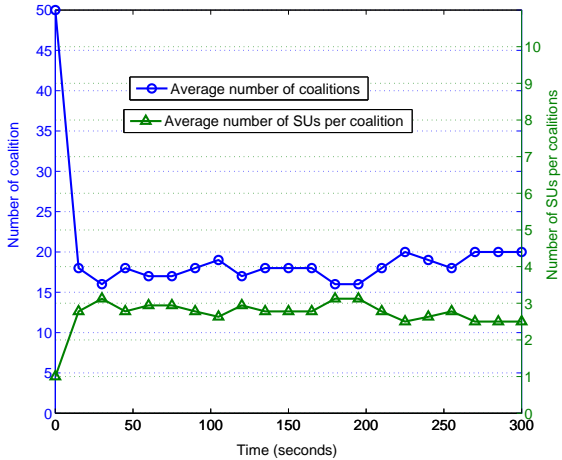}
\end{center}\vspace{-0.5cm}
\caption {Evolution of the network's structure over time using CF for a network of $N=50$~SUs, a constant speed of $120$~km/h and a non-cooperative false alarm $P_f=0.01$.} \label{fig:mobconv}\vspace{-0.4cm}
\end{figure}
In Fig.~\ref{fig:win}, we evaluate the performance of the CF-PD algorithm of Table~\ref{tab:alg2} by showing the average percentage of winning SUs (average over locations of SUs and non-cooperative false alarm range $P_f \le \alpha$ ), i.e., SUs achieving a detection probability $\chi=95\%$ imposed by the PU (the choice of $\chi=95\%$ conforms with the recommendation of the IEEE 802.22 standard \cite{RE00}) as a function of
the network size $N$.  We compare the performance of CF-PD to that of the non-cooperative case  as well as the optimal partition that maximizes the number of winning SUs and which is found by a centralized search over all partitions formed of MWCs. For CF-PD, the percentage of winning SUs increases with the number of SUs since the possibility of finding cooperating
partners increases. In contrast, the non-cooperative approach presents an almost constant performance with different network sizes. Cooperation presents a significant
advantage over the non-cooperative case in terms of average percentage of winning SUs, and this advantage increases with the network size approximately tripling the number of winning SUs (relative to the non-cooperative case) at $N=50$.  Furthermore, compared to the optimal solution, the proposed CF-PD algorithm achieves a highly comparable performance with a loss not exceeding $3.7$ percentage points at $N=7$. This shows that, by using the proposed CF-PD algorithm, the network can achieve a performance that is close to optimal. Note that, for networks larger than $N=7$ SUs, finding the optimal partition through exhaustive search is mathematically and computationally intractable. However, as the network size increases CF-PD presents a portion of winning SUs that is close to $100\%$, reaching up to $87.25\%$ at $N=50$, hence, as the network grows, the performance gap relatively to the optimal case will not increase much. Finally, in Fig.~\ref{fig:win}, we note that for the proposed algorithm, the non-cooperative case, and the optimal solution, some SUs (or coalitions)  are still losing, i.e., unable to meet the $\chi=95\%$ constraint. This is mainly due to the current locations and  channels of these SUs. In order for these SUs to meet the desired performance level, they need to wait for environmental changes, such as their own mobility (moving to a better location) or the arrival of new SUs that can boost their performance.
\begin{figure}[!t]
\begin{center}
\includegraphics[angle=0,width=80mm]{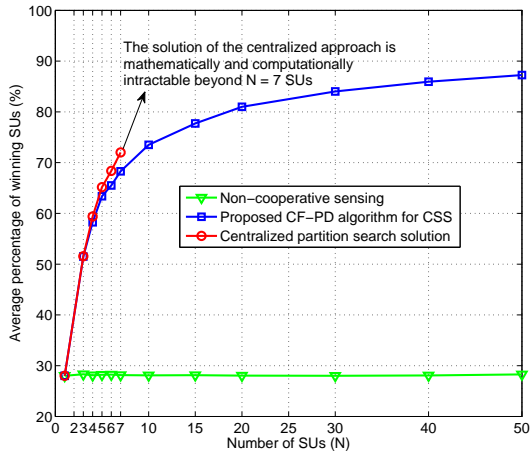}
\end{center}\vspace{-0.5cm}
\caption {Average percentage of winning SUs achieving $\chi=95\%$ detection probability versus number of SUs $N$.} \label{fig:win}\vspace{-0.4cm}
\end{figure}

In Fig.~\ref{fig:wingamm}, we show the average percentage of winning SUs (average over locations of SUs and non-cooperative false alarm range $P_f \le \alpha$ ) achieved by CF-PD for $N=50$~SUs as the desired probability of miss $\gamma=1-\chi$ increases (i.e. desired detection probability decreases) within the range recommended by the IEEE 802.22 \cite{RE00}. As $\gamma$ increases, the number of winning SUs increases for both the CF-PD and the non-cooperative case, as it becomes easier to guarantee this required probability. Fig.~\ref{fig:wingamm} shows that the proposed CF-PD algorithm guarantees that at least $74.1\%$ of the SUs (at $\gamma=0.01$, the most stringent constraint) achieve the desired detection probability. Further, at all $\gamma$, CF-PD presents a significant advantage over the non-cooperative case which  increases as the desired detection probability becomes stricter. For instance, at $\gamma=0.01$ (i.e. $\chi=99\%$), CF-PD  allows around eight times more SUs (relative to the non-cooperative case) to achieve the required probability of detection.
\begin{figure}[!t]
\begin{center}
\includegraphics[angle=0,width=80mm]{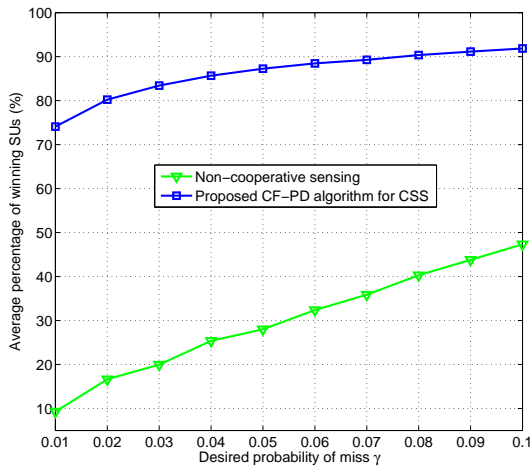}
\end{center}\vspace{-0.5cm}
\caption {Average percentage of winning SUs versus the desired probability of miss $\gamma=1-\chi$ for the proposed CF-PD algorithm for $N=50$~SUs.} \label{fig:wingamm}\vspace{-0.4cm}
\end{figure}
In Fig.~\ref{fig:numus}, for a network of $N=50$ SUs, we evaluate the sizes of the coalitions resulting from CF and CF-PD (with $\chi=95\%$) and compare them with the the upper bound $M_{\max}$ derived in Theorem~2. First, as the non-cooperative $P_f$ increases, for CF and CF-PD, both the maximum and the average size of the formed coalitions decrease converging towards the non-cooperative case as $P_f$ reaches the constraint $\alpha = 0.1$. Through this result, we can clearly see the limitations that the detection-false alarm probabilities trade off for CSS imposes on the coalition size and network structure. Also, in Fig.~\ref{fig:numus}, we show that, albeit the upper bound on coalition size $M_{\max}$ increases drastically as $P_f$ becomes smaller, the average maximum coalition size achieved by the proposed algorithms does not exceed $5$ SUs per coalition for CF and $3$ SUs per coalition for CF-PD. This result shows that the network structure is composed of a large number of small coalitions rather than a small number of large coalitions, especially for CF-PD, even when $P_f$ is small and the collaboration possibilities are high. In this context, for CF-PD, the average
and average maximum coalition sizes are smaller than for CF, since in CF-PD, once a coalition achieves the detection probability $\chi$, no further collaboration is required.

In Fig.~\ref{fig:speed}, for CF and CF-PD (with $\chi=95\%$) we show, over a period of $10$ minutes, the frequency in terms of merge-and-split and adjust operations per minute for various speeds of the SUs in a mobile cognitive network with $N=50$~SUs and $P_f=0.01$. As the velocity increases,  the frequency of all operations
increases for both CF and CF-PD due to the changes in the network structure incurred by
mobility. The CF algorithm incurs a higher frequency of merge-and-split than CF-PD since, in the CF, the SUs seek to minimize their probabilities of miss with no constraints. In contrast, the CF-PD presents a significantly smaller frequency of merge-and-split operations, since the SUs, once they form a winning coalition, are no longer interested in participating in the merge-and-split operations. However, the CF-PD complements this lower number of merge-and-split operations with a number of adjust operations, needed in order to form MWCs out of winning coalitions. Finally, the frequency of all operations (merge, split and adjust) for CF-PD is comparable to the frequency of merge-and-split in CF at almost all velocities. However, CF-PD possesses a lower complexity as the adjust rule has a complexity smaller than merge or split.
%, allowing coalitions of virtually any size to form. However, despite the possibility of forming very large coalitions, we note that t

\begin{figure}[!t]
\begin{center}
\includegraphics[angle=0,width=80mm]{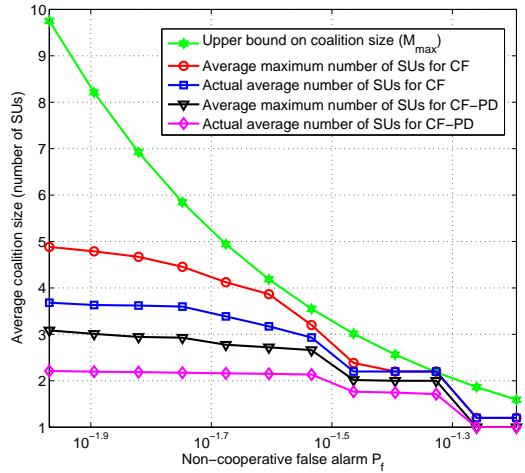}
\end{center}\vspace{-0.5cm}
\caption {Maximum and average coalition size versus non-cooperative false alarm $P_f$ (or energy detection threshold $\lambda$) for CF and CF-PD for $N=50$~SUs.} \label{fig:numus}\vspace{-0.4cm}
\end{figure}
\begin{figure}[!t]
\begin{center}
\includegraphics[angle=0,width=80mm]{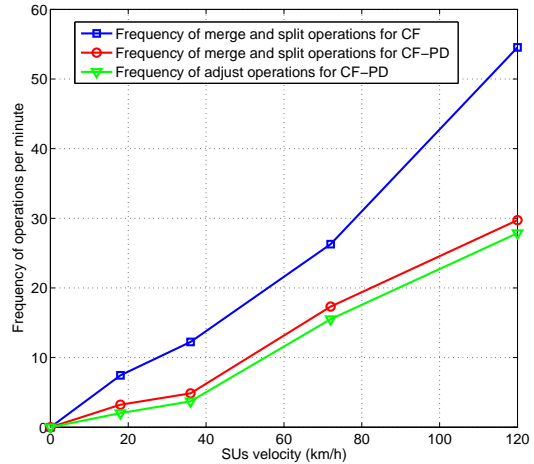}
\end{center}\vspace{-0.5cm}
\caption {Frequency of merge-and-split operations per minute (for CF and CF-PD) and adjust operations per minute (for CF-PD) achieved over a period of $10$ minutes for different velocities for $N=50$~SUs and $P_f=0.01$.}\vspace{-0.4cm} \label{fig:speed}
\end{figure}

\section{Conclusions}\label{sec:conc}
 This paper introduced a novel and distributed model for performing collaborative sensing in cognitive radio networks. In the proposed model, a network of SUs can interact and form cooperating coalitions for improving their spectrum sensing performance. We modeled the problem as a coalitional game with non-transferable utility and derived distributed algorithms for coalition formation. First, we proposed a distributed coalition formation  (CF) algorithm based on two rules of merge and split that enable SUs in a cognitive network to collaborate and maximize the probability of detecting the PU presence while accounting for the costs in terms of false alarm probability. We characterized the network structure resulting from the proposed algorithm as well as studied and proved its various properties. We further complemented the proposed algorithm with an adjunct coalitional voting game for providing distributed coalition formation with detection probability guarantees (CF-PD) when such detection probabilities are imposed by the PU. Simulation results showed that the CF algorithm reduces the average probability of miss per SU up to $88.45\%$ compared to the non-cooperative case while the CF-PD algorithm enables up to $87.25\%$ of the SUs to achieve the required detection probability while forming minimal winning coalitions. The results also showed how, through the proposed algorithms, the SUs can self-organize and adapt the network structure to environmental changes such as mobility.
\appendix
\subsection{Proof of Theorem 1}
Consider two SUs having non-cooperative probabilities of miss $P_{m,i}$ and $P_{m,j}$, respectively, such that $P_{m,i} \le P_{m,j}$. In this case, if and when coalition $S=\{i,j\}$ forms, then SU $i$ is the \emph{coalition head} as per Convention 1 and we have
\begin{align}
Q_{m,S} \! =\! P_{m,i}[ P_{m,j}(1-P_{e,j,i})\!  +\!  (1- P_{m,j})P_{e,j,i}],\\
Q_{f,S} \! =\!  1 - \left[\displaystyle (1-P_f)[ (1-P_{f})(1-P_{e,j,i}) \! + \! P_{f}P_{e,j,i}]\right].
\end{align}
%The corresponding utilities for $\{1\}$, $\{2\}$ and coalition $S$ are (using \ref{eq:util})
%\begin{equation}\label{eq:util1}
%v(\{1\}) =  (1 - P_{m,1}) + ( \alpha^2\cdot\log{\left(1-\left(\frac{P_f}{\alpha}\right)^2\right)}),
%\end{equation}
%\begin{equation}\label{eq:util2}
%v(\{2\}) =  (1 - P_{m,2}) + ( \alpha^2\cdot\log{\left(1-\left(\frac{P_f}{\alpha}\right)^2\right)}),
%\end{equation}
%\begin{equation}\label{eq:util12}
%v(S) =  (1 - Q_{m,S}) +  \alpha^2\cdot\log{\left(1-\left(\frac{Q_{f,S}}{\alpha}\right)^2\right)},
%\end{equation}
By virtue of Property~\ref{prop:equ}, the payoffs of the SUs $i$ and $j$ if $S$ forms are given by, $\phi_i(S) = \phi_j(S) = v(S)$. In order for the merge between SUs $i$ and $j$ to occur (for forming $S$) we must have $S \rhd \{\{i\},\{j\}\}$. Using the Pareto order, and given that $v(\{i\}) \ge v(\{j\})$, $S \rhd \{\{i\},\{j\}\}$ translates into
%\begin{align}\label{eq:paretocond}
$\phi_i(S) = v(S) > v(\{i\}) \ge v(\{j\})$.
%\end{align}
After algebraic manipulations and using the expressions of $v(\{i\})$ and $v(S)$ in (\ref{eq:util}), we obtain
\begin{align}\label{eq:trans}
\frac{D\cdot P_{e,j,i} +C}{\alpha^2} \le \log{\left(1-\left(\frac{A\cdot P_{e,j,i} + B}{\alpha}\right)^2\right)},
\end{align}
with $P_{e,j,i}$ the probability of error between the two SUs, and $A\!=\!(2P_f -1)(P_f - 1)$, $B\!\!=\!\! P_f(2-P_f)$, $C\!\!=\!\!P_{m,i}(P_{m,j}-1)+\left(\alpha^2\cdot\log{\left(1-\left(\frac{P_f}{\alpha}\right)^2\right)}\right)$ and $D\! =\!P_{m,i}(1-2P_{m,j})$.

The inequality in (\ref{eq:trans}) is a transcendental inequality that admits only a numerical (graphical) solution, with no analytical expression \cite{LZ00}. The numerical solution (by solution we imply the range of $P_{e,j,i}$ that satisfy the inequality) can be derived by plotting both sides of the equation as a function of $P_{e,j,i}$. A sufficient approximate solution (i.e. approximate range of $P_{e,j,i}$ that satisfy the inequality) for (\ref{eq:trans}) is found by solving

\begin{align}\label{eq:trans2}
\frac{D\cdot \bar{P}_{e,j,i} +C}{\alpha^2} \le \log{\left(1-\left(\frac{A\cdot P_{e,j,i} + B}{\alpha}\right)^2\right)},
\end{align}
\noindent where $\bar{P}_{e,j,i}$ is the value that maximizes $\frac{D\cdot P_{e,j,i} +C}{\alpha^2} $ (since the probability of error is bounded, this maximum is finite) and which directly depends on the sign of $D$. By algebraic manipulation, (\ref{eq:trans2}) yields a quadratic inequality in $P_{e,j,i}$
\begin{align}\label{eq:fina2}
A^2P_{e,j,i}^2 + (2AB)P_{e,j,i}+ B^2 + \alpha^2(e^{(\frac{C+\bar{P}_{e,j,i} D}{\alpha^2})}-1) \le 0,
\end{align}
with discriminant $\Delta = A^2\alpha^2(1-e^{(\frac{C+\bar{P}_{e,j,i}D}{\alpha^2})})$. We distinguish two cases depending on the sign of $D\! =\!P_{m,i}(1-2P_{m,j})$.

\textbf{Case 1}:  $D\ge 0 $, that is $P_{m,j} \le \frac{1}{2}$. First, we note that: (i) $P_f$ which depends solely on the detection threshold $\lambda$ is generally chosen relatively small compared to $1$ (a typical value is $0.01$), and (ii) the inequality (\ref{eq:trans}) admits a solution if and only if $A\cdot P_{e,j,i} + B < \alpha$ (otherwise the right hand side of (\ref{eq:trans}) goes to $-\infty$). A direct result of (ii) is that any solution $\tilde{P}_{e,j,i}$ for (\ref{eq:trans}) cannot exceed $\frac{\alpha - B}{A}$. By using (i), we have that $A \approx 1$ and $B \approx 0$, hence, any solution $\tilde{P}_{e,j,i}$ can be approximately bounded by $\tilde{P}_{e,j,i} \le \alpha$, and thus, the probability of error $\bar{P}_{e,j,i} = \alpha$. In this case, $\Delta$ is always positive. After manipulation $(1-e^{(\frac{C+\alpha D}{\alpha^2})}) > 0$ yields
\setlength{\aligntop}{-0.8em}
\setlength{\alignbot}{-1.25\baselineskip}
\begin{align}
P_{m,i}(1\!-\!P_{m,j})\! +\! P_{m,i}(2P_{m,j}-1)\!  >\! \alpha^2\log{(\!1\!-\!\frac{P_f^2}{\alpha^2})}.
\end{align}
\setlength{\aligntop}{-0.1em}
\setlength{\alignbot}{-0.7\baselineskip}
\noindent The RHS is always negative (barrier function \cite{BO00}). The LHS can be shown to be always positive since otherwise, we need $P_{m,i}(1-P_{m,j}) + \alpha P_{m,i}(2P_{m,j}-1) <0$ which yields two cases
%\begin{align}
$P_{m_j} > \frac{\alpha - 1}{2\alpha - 1} \ge 1\ \textrm{if } \alpha < 0.5$ or
$P_{m_j} < \frac{\alpha - 1}{2\alpha - 1} \le 0\ \textrm{if } \alpha > 0.5$
%\end{align}
that are both impossible for $0 \le P_{m,j} \le 1$ (probability of miss). Hence, when $P_{m,j} \le \frac{1}{2}$, (\ref{eq:trans2}) admits two roots
\setlength{\aligntop}{-0.6em}
\setlength{\alignbot}{-1.25\baselineskip}
\begin{align}
\tilde{P}_{e,j,i} = \frac{(- B + \alpha \sqrt{(1-e^{(\frac{C+\alpha D}{\alpha^2})} )})}{A},\nonumber\\
\tilde{P}_{e,j,i}^{'}  = \frac{(- B - \alpha \sqrt{(1-e^{(\frac{C+\alpha D}{\alpha^2})} )})}{A}.
\end{align}
\setlength{\aligntop}{-0.1em}
\setlength{\alignbot}{-0.7\baselineskip}
Or $\tilde{P}_{e,j,i}^{'} < 0$ ($A>0$ due to (i)). Consequently, the solution of (\ref{eq:fina2}) is $P_{e,j,i} \le \tilde{P}_{e,j,i}$ where $\tilde{P}_{e,j,i}$ is given by (after algebraic manipulation and substituting $A$,$B$,$C$, and $D$ for their values)
\setlength{\aligntop}{-0.5em}
\setlength{\alignbot}{-1.25\baselineskip}
\begin{align}\label{eq:finish}
  \tilde{P}_{e,j,i} = \frac{ P_f(P_f -2) + \alpha \sqrt{(1-(1-\frac{P_f^2}{\alpha^2})e^{\frac{P_{m,i}(1-\alpha)(P_{m,j}-1) }{\alpha^2}} )}}{(2P_f -1)(P_f - 1)}.
\end{align}
\setlength{\aligntop}{-0.1em}
\setlength{\alignbot}{-0.7\baselineskip}
The solution $P_{e,j,i} \le \tilde{P}_{e,j,i}$ with $\tilde{P}_{e,j,i}$ in (\ref{eq:finish}) provides a sufficient analytical approximation for the solution (range of probabilities) of the initial transcendental inequality in (\ref{eq:trans}) when $P_{m,j} \le \frac{1}{2}$.

\textbf{Case 2}: $D < 0 $, that is $P_{m,j} > \frac{1}{2}$. In this case, we can easily see that $\bar{P}_{e,1,2} = 0$, and it is easily seen that the discriminant $\Delta$ is positive. By following a reasoning analogous to Case 1, we find that the sufficient analytical approximation is $P_{e,j,i} \le \tilde{P}_{e,j,i}$ with $ \tilde{P}_{e,j,i}$ given by
\setlength{\aligntop}{-0.7em}
\setlength{\alignbot}{-1.15\baselineskip}
\begin{align}
\tilde{P}_{e,j,i} = \frac{ P_f(P_f -2) + \alpha \sqrt{(1-(1-\frac{P_f^2}{\alpha^2})e^{\frac{P_{m,i}(P_{m,j}-1) }{\alpha^2}} )}}{(2P_f -1)(P_f - 1)}.
\end{align}
\setlength{\aligntop}{-0.1em}
\setlength{\alignbot}{-0.7\baselineskip}
We cam combine these two cases into one solution, that is $P_{e,j,i} \le \tilde{P}_{e,j,i}$ with
\setlength{\aligntop}{-0.4em}
\setlength{\alignbot}{-1.15\baselineskip}
\begin{align}\label{eq:finishall}
\tilde{P}_{e,j,i}\! =\! \frac{ P_f(P_f -2)\! +\! \alpha \sqrt{(1-(1-\frac{P_f^2}{\alpha^2})e^{\frac{P_{m,i}(1-\eta)(P_{m,j}-1) }{\alpha^2}} )}}{(2P_f -1)(P_f - 1)},
\end{align}
\setlength{\aligntop}{-0.1em}
\setlength{\alignbot}{-0.7\baselineskip}
with $ \eta = \alpha$, if $P_{m,j} \le \frac{1}{2}$ and $\eta= 0$ otherwise.

\nocite{WS02}

\def\baselinestretch{0.88}
\bibliographystyle{IEEEtran}
\bibliography{references}

\end{document}